\documentclass[11pt]{article}

\usepackage{xspace}
\usepackage[in]{fullpage}
\usepackage{graphicx}
\usepackage{cite,url}
\usepackage{amsthm}
\usepackage{amsfonts}
\usepackage{tweaklist}

\newtheorem{theorem}{Theorem}[section]

\newtheorem{cor}[theorem]{Corollary}
\newtheorem{obs}[theorem]{Observation}
\newtheorem{problem}[theorem]{Problem}
\newtheorem{lemma}[theorem]{Lemma}

\newtheorem{claim}[theorem]{Claim}

\theoremstyle{definition}
\newtheorem{definition}[theorem]{Definition}

\theoremstyle{remark}

\newtheorem{example}[theorem]{Example}

\newcommand{\ie}{\emph{i.e.}}

\newcommand{\eg}{\emph{e.g.}}
\newcommand{\etc}{\emph{etc}}
\newcommand{\etal}{\emph{et al}}

\newcommand{\iod}{\ensuremath{\mathsf{IoD}}}
\newcommand{\run}[1]{\ensuremath{\mathbf{#1}}}
\newcommand{\ex}[2]{\ensuremath{\left<#1;#2\right>}}

\renewcommand{\P}{\ensuremath{\mathcal{P}}}
\newcommand{\Runs}{\ensuremath{\mathcal{R}}}

\newcommand{\N}{\ensuremath{\mathbb{N}}}
\newcommand{\R}{\ensuremath{\mathbb{R}}}

\renewcommand{\enumhook}{\setlength{\topsep}{0pt}%
  \setlength{\itemsep}{0pt}}
\renewcommand{\itemhook}{\setlength{\topsep}{0pt}%
  \setlength{\itemsep}{0pt}}
\renewcommand{\descripthook}{\setlength{\topsep}{0pt}%
  \setlength{\itemsep}{0pt}}

\begin{document}

\date{}
\title{Distributed Computing with Adaptive Heuristics\\
{\small Revised version will appear in the \emph{Proceedings of Innovations in Computer Science 2011}}}

\author{Aaron D.\ Jaggard
\thanks{Supported in part by
NSF grants 0751674 and 0753492.}\\
Dept. of Computer Science, Colgate University\\ DIMACS, Rutgers University\\
\texttt{adj@dimacs.rutgers.edu}
\and Michael Schapira\thanks{Supported by NSF grant 0331548.}\\
Dept. of Computer Science\\
Yale University and UC Berkeley\\
\texttt{michael.schapira@yale.edu}
\and Rebecca N.\ Wright
\thanks{Supported in part by NSF grant 0753061.}\\
Dept.\ of Computer Science and DIMACS\\
Rutgers University\\
\texttt{rebecca.wright@rutgers.edu}}

\maketitle

\setcounter{page}{0}
\thispagestyle{empty}

\begin{abstract}
We use ideas from distributed computing to study dynamic environments
in which computational nodes, or decision makers, follow \emph{adaptive heuristics}~\cite{H05}, \ie, simple and
unsophisticated rules of behavior, \emph{e.g.}, repeatedly ``best replying'' to others' actions, and minimizing ``regret'',
that have been extensively studied in game theory and economics. We explore when convergence of such simple dynamics
to an equilibrium is guaranteed in asynchronous computational
environments, where nodes can act at any time. Our research agenda,
\emph{distributed computing with adaptive heuristics}, lies on the
borderline of computer science (including distributed computing and
learning) and game theory (including game dynamics and adaptive
heuristics). We exhibit a general non-termination result for a broad
class of heuristics with bounded recall---that is, simple rules
of behavior that depend only on recent history of interaction between
nodes. We consider implications of our result across a wide variety of
interesting and timely applications: game theory, circuit design,
social networks, routing and congestion control. We also study the
computational and communication complexity of asynchronous dynamics
and present some basic observations regarding the effects of
asynchrony on no-regret dynamics. We believe that our work opens a new
avenue for research in both distributed computing and game theory.
\end{abstract}

\newpage

\section{Introduction}

Dynamic environments where computational nodes, or decision makers,
repeatedly interact arise in a variety of settings, such as Internet
protocols, large-scale markets, social networks, multi-processor
computer architectures, and more. In many such settings, the
prescribed behavior of the nodes is often simple, natural and myopic
(that is, a heuristic or ``rule of thumb''), and is also adaptive, in
the sense that nodes constantly and autonomously react to others. These ``adaptive
heuristics''---a term coined in~\cite{H05}---include simple
behaviors, \emph{e.g.}, repeatedly ``best replying'' to others' actions, and minimizing ``regret'',
that have been extensively studied in game theory and economics.

Adaptive heuristics are simple and unsophisticated, often
reflecting either the desire or necessity for computational nodes
(whether humans or computers) to provide quick responses and have a
limited computational burden. In many interesting
contexts, these adaptive heuristics can, in the long run, move the
global system in good directions and yield highly rational and
sophisticated behavior, such as in game theory results demonstrating
the convergence of best-response or no-regret dynamics to equilibrium
points (see~\cite{H05} and references therein).

However, these positive results for adaptive heuristics in game theory are, with but
a few exceptions (see Section~\ref{sec:related}), based on the sometimes implicit
and often unrealistic premise that nodes' actions are somehow
synchronously coordinated. In many settings, where nodes can act at
any time, this kind of synchrony is not available.  It has long been
known that asynchrony introduces substantial difficulties in
distributed systems, as compared to synchrony~\cite{FLP85}, due to the
``limitation imposed by local knowledge''~\cite{Lyn89}. There has been
much work in distributed computing on identifying conditions that
guarantee protocol termination in asynchronous computational
environments. Over the past three decades, we have seen many results
regarding the possibility/impossibility borderline for
failure-resilient computation~\cite{Lyn89,FR03}. In the classical
results of that setting, the risk of non-termination stems from the
possibility of failures of nodes or other components.

We seek to bring together these two areas to form a new
research agenda on distributed computing with adaptive heuristics. Our aim is to
draw ideas from distributed computing theory to investigate provable
properties and possible worst-case system behavior of adaptive
heuristics in asynchronous computational environments. We take the first steps of this research agenda.  We show that a large
and natural class of adaptive heuristics fail to provably converge to an
equilibrium in an asynchronous setting, even if the nodes and communication channels are
guaranteed to be failure-free. This has implications across a wide domain of applications: convergence of game dynamics to pure Nash equilibria; stabilization of asynchronous circuits; convergence to a stable routing tree of the Border Gateway Protocol, that handles Internet routing; and more. We also explore the impact of scheduling on convergence
guarantees.  We show that non-convergence is not inherent to adaptive
heuristics, as some forms of regret minimization provably converge in
asynchronous settings.  In more detail, we make the following
contributions:

\vspace{0.02in}\noindent{\bf General non-convergence result (Section~\ref{sec:bounded}).}
It is often desirable or necessary due to practical constraints that
computational nodes' (\eg, routers') behavior rely on limited memory
and processing power. In such contexts, nodes' adaptive heuristics are
often based on \emph{bounded recall}---\ie, depend solely on recent
history of interaction with others---and can even be
\emph{historyless}---\ie, nodes only react to other nodes' current
actions). We exhibit a general impossibility result using a valency
argument---a now-standard technique in distributed computing
theory~\cite{Lyn89, FR03}---to show that a broad class of bounded-recall adaptive heuristics cannot always converge to a stable
state. More specifically, we show that, for a large family of such
heuristics, simply the existence of two ``equilibrium points'' implies that
there is some execution that does not converge to any outcome even if nodes and
communication channels are guaranteed not to fail. We also give
evidence that our non-convergence result is essentially tight.

\vspace{0.02in}\noindent{\bf Implications across a wide variety of
interesting and timely applications (Section~\ref{sec:examples}).} We
apply our non-convergence result to a wide variety of interesting
environments, namely convergence of game dynamics to pure Nash
equilibria, stabilization of asynchronous circuits, diffusion of
technologies in social networks, routing on the Internet, and
congestion control protocols.

\vspace{0.02in}\noindent{\bf Implications for convergence of $r$-fairness and randomness
(Section~\ref{sec:schedules}).} We study the effects on convergence to
a stable state of natural restrictions on the order of nodes'
activations (i.e., the order in which nodes' have the opportunity to
take steps), that have been extensively studied in distributed
computing theory: (1) \emph{$r$-fairness}, which is the guarantee that each node
selects a new action at least once within every $r$ consecutive time steps, for some pre-specified $r>0$; and (2) randomized selection of the initial state of the system and the order of nodes' activations.

\vspace{0.02in}\noindent{\bf Communication and computational complexity of asynchronous
dynamics (Section~\ref{sec:complexity}).} We study the tractability of
determining whether convergence to a stable state is guaranteed. We
present two complementary hardness results that establish that, even
for extremely restricted kinds of interactions, this feat is hard: (1)
an exponential communication complexity lower bound; and (2) a
computational complexity PSPACE-completeness result that, alongside
its computational implications, implies that we cannot hope to have
short witnesses of guaranteed asynchronous convergence (unless PSPACE
$\subseteq$ NP).

\vspace{0.02in}\noindent{\bf Asynchronous no-regret dynamics (Section~\ref{sec:regret}).}
We present some basic observations about the convergence properties of
no-regret dynamics in our framework, that establish that, in contrast
to other adaptive heuristics, regret minimization is quite robust to
asynchrony.

\vspace{0.02in}\noindent{\bf Further discussion of a research agenda in distributed
computing with adaptive heuristics (Section~\ref{sec:future})} We believe that this work has but scratched the
surface in the exploration of the behavior of adaptive heuristics in asynchronous computational environments. Many
important questions remain wide open. We present context-specific problems in the relevant sections, and also outline general interesting directions for future research in Section~\ref{sec:future}.\vspace{0.02in}

Before presenting our main results, we overview related work
(Section~\ref{sec:related}) and provide a detailed description of our
model (Section~\ref{sec:model}).

\section{Related Work} \label{sec:related}

Our work relates to many ideas in game theory and in distributed computing.  We discuss game theoretic work on adaptive heuristics and on asynchrony, and also distributed computing work on fault tolerance and self stabilization. We also highlight the application areas we consider.

\vspace{0.02in}\noindent{\bf Adaptive heuristics.} Much work in game
theory and economics deals with adaptive heuristics (see Hart~\cite{H05}
and references therein). Generally speaking, this long line of research
explores the ``convergence'' of simple and myopic rules of behavior (\emph{e.g.}, best-response/fictitious-play/no-regret dynamics) to
an ``equilibrium''. However, with few exceptions (see below), such analysis has so far primarily concentrated on
synchronous environments in which steps take place simultaneously or in some other predetermined
prescribed order. In contrast, we explore adaptive heuristics in asynchronous environments, which are more realistic for many applications.

\vspace{0.02in}\noindent{\bf Game-theoretic work on asynchronous environments.} Some game-theoretic work on repeated games considers ``asynchronous moves''.\footnote{Often, the term asynchrony merely indicates that players are not all activated at each time step, and thus is used to describe environments where only one player is activated at a time (``alternating moves''), or where there is a probability distribution that determines who is activated when.} (see~\cite{LM97,Y04}, among others, and references therein). Such work does not explore the behavior of dynamics, but has other research goals (\emph{e.g.}, characterizing equilibria, establishing Folk theorems). We are, to the best of our knowledge, the first to study the effects of asynchrony (in the broad distributed computing sense) on the convergence of \emph{game dynamics} to equilibria.

\vspace{0.02in}\noindent{\bf Fault-tolerant computation.} We use ideas and techniques from work in distributed
computing on protocol termination in asynchronous computational environments where nodes and communication channels are possibly faulty. Protocol termination in such environments, initially motivated by multi-processor computer architectures, has been extensively studied
in the past three decades~\cite{FLP85, Ben-Or86, DDS87, BG93, HS99,SZ00}, as nicely surveyed in~\cite{Lyn89,FR03}. Fischer, Lynch and Paterson~\cite{FLP85} showed, in a landmark paper, that a broad class of failure-resilient consensus protocols cannot provably terminate. Intuitively, the risk of protocol nontermination in~\cite{FLP85} stems from the possibility of failures; a computational node cannot tell whether another node is silent due to a failure or is simply taking a long time to react. Our focus here is, in contrast, on failure-free environments.

\vspace{0.02in}\noindent{\bf Self stabilization.} The concept of self stabilization is fundamental to distributed computing and dates back to Dijkstra, 1973 (see~\cite{D00} and references therein). Convergence of adaptive heuristics to an ``equilibrium'' in our model can be viewed as the self stabilization of such dynamics (where the ``equilibrium points'' are the legitimate configurations). Our formulation draws ideas from work in distributed computing (\emph{e.g.}, Burns' distributed daemon model) and in networking research~\cite{GSW02} on self stabilization.

\vspace{0.02in}\noindent{\bf Applications.} We discuss the implications of our non-convergence result across a wide variety of
applications, that have previously been studied: convergence of game dynamics (see, \eg,~\cite{HM03,HM06}); asynchronous circuits (see,
\eg,~\cite{DN97}); diffusion of innovations, behaviors, \etc., in social networks (see Morris~\cite{Morris00} and also~\cite{IKMW07}); interdomain routing~\cite{GSW02,SSZ09}; and congestion
control~\cite{GSZS09}.

\section{The Model}\label{sec:model}
We now present our model for analyzing adaptive heuristics in asynchronous environments.

\vspace{0.02in} \noindent {\bf Computational nodes interacting.} There is an \emph{interaction system} with $n$ \emph{computational nodes}, $1,\ldots,n$. Each computational node $i$ has an \emph{action space} $A_i$. Let $A=\times_{j\in [n]}A_j$, where $[n]=\{1,\ldots,n\}$. Let $A_{-i}=\times_{j\in [n]\setminus\{i\}}A_j$. Let $\Delta(A_i)$ be the set of all probability distributions over the actions in $A_i$.

\vspace{0.02in}\noindent {\bf Schedules.} There is an infinite sequence of discrete time steps $t=1,\ldots$.  A \emph{schedule} is a function $\sigma$ that maps each $t\in \N_{+}=\{1,2,\ldots\}$ to a nonempty set of computational nodes: $\sigma(t)\subseteq [n]$.  Informally, $\sigma$ determines (when we consider the dynamics of the system) which nodes are \emph{activated} in each time-step.  We say that a schedule $\sigma$ is \emph{fair} if each node $i$ is activated infinitely many times in $\sigma$, \ie, $\forall i\in [n]$, there are infinitely many $t\in\N_{+}$ such that $i\in \sigma(t)$. For $r\in\N_+$, we say that a schedule $\sigma$ is \emph{$r$-fair} if each node is activated at least once in every sequence of $r$ consecutive time steps, \ie, if, for every $i\in[n]$ and $t_0\in\N_+$, there is at least one value $t\in\{t_0,t_0+1,\ldots,t_0+r-1\}$ for which $i\in \sigma(t)$.

\vspace{0.02in}\noindent {\bf History and reaction functions.} Let $H_{0}=\emptyset$, and let $H_{t}=A^t$ for every $t\geq 1$. Intuitively, an element in $H_{t}$ represents a possible \emph{history} of interaction at time step $t$.  For each node $i$, there is an infinite sequence of functions $f_i = (f_{(i,1)},f_{(i,2)},\ldots,f_{(i,t)},\ldots)$ such that, for each $t\in\N_{+}$, $f_{(i,t)}:H_t\rightarrow \Delta(A_i)$; we call $f_i$ the \emph{reaction function} of node $i$.  As discussed below, $f_i$ captures $i$'s way of responding to the history of interaction in each time step.

\vspace{0.02in}\noindent {\bf Restrictions on reaction functions.} We now present five possible restrictions on reaction functions: determinism, self-independence, bounded recall, stationarity and historylessness.

\begin{enumerate}

\item {\bf Determinism:} a reaction function $f_i$ is \emph{deterministic} if, for each input, $f_i$ outputs a single action (that is, a probability distribution where a single action in $A_i$ has probability $1$).

\item {\bf Self-independence:} a reaction function $f_i$ is \emph{self-independent} if node $i$'s own (past and present) actions do not affect the outcome of $f_i$. That is, a reaction function $f_i$ is self-independent if for every $t\geq 1$ there exists a function $g_t:A_{-i}^t\rightarrow \Delta(A_i)$ such that $f_{(i,t)}\equiv g_t$.

\item {\bf $k$-recall and stationarity:} a node $i$ has \emph{$k$-recall} if its reaction function $f_i$ only depends on the $k$ most recent time steps, \ie, for every $t\geq k$, there exists a function $g:H_{k}\rightarrow \Delta(A_i)$ such that $f_{(i,t)}(x)=g(x_{|k})$ for each input $x\in H_t$ ($x_{|k}$ here denotes the last $k$ coordinates, \ie, $n$-tuples of actions, of $x$). We say that a $k$-recall reaction function is \emph{stationary} if the time counter $t$ is of no importance. That is, a $k$-recall reaction function is stationary if there exists a function $g:H_{k}\rightarrow \Delta(A_i)$ such that \emph{for all} $t\geq k$, $f_{(i,t)}(x)=g(x_{|k})$ for each input $x\in H_t$.

\item {\bf Historylessness:} a reaction function $f_i$ is \emph{historyless} if $f_i$ is $1$-recall \emph{and} stationary, that is, if $f_i$ only depends on $i$'s and on $i$'s neighbors' most recent actions.

\end{enumerate}

\vspace{0.02in}\noindent {\bf Dynamics.} We now define \emph{dynamics} in our model. Intuitively, there is some initial state (history of interaction) from which the interaction system evolves, and, in each time step, some subset of the nodes reacts to the past history of interaction. This is captured as follows. Let $s^{(0)}$, that shall be called the ``\emph{initial state}'', be an element in $H_w$, for some positive $w\in\N$. Let $\sigma$ be a schedule. We now describe the ``\emph{$(s^{(0)},\sigma)$-dynamics}''. The system's evolution starts at time $t=w+1$, when each node $i\in \sigma(w+1)$ simultaneously chooses an action according to $f_{(i,w+1)}$, \ie, node $i$ randomizes over the actions in $A_i$ according to $f_{(i,w+1)}(s^{(0)})$. We now let $s^{(1)}$ be the element in $H^{w+1}$ for which the first $w$ coordinates ($n$-tuples of nodes' actions) are as in $s^{(0)}$ and the last coordinate is the $n$-tuple of \emph{realized} nodes' actions at the end of time step $t=w+1$. Similarly, in each time step $t>w+1$, each node in $\sigma(t)$ updates its action according to $f_{(i,t)}$, based on the past history $s^{(t-w-1)}$, and nodes' realized actions at time $t$, combined with $s^{(t-w-1)}$, define the history of interaction at the end of time step $t$, $s^{(t-w)}$.

\vspace{0.02in}\noindent {\bf Convergence and convergent systems.} We say that nodes' actions \emph{converge} under the $(s^{(0)},\sigma)$-dynamics if there exist some positive $t_0\in\N$, and some action profile $a=(a_1,\ldots,a_n)$, such that, for all $t>t_0$, $s^{(t)}=a$. The dynamics is then said to converge to $a$, and $a$ is called a ``\emph{stable state}'' (for the $(s^{(0)},\sigma)$-dynamics), \ie, intuitively, a stable state is a global action state that, once reached, remains unchanged. We say that the interaction system is \emph{convergent} if, for all initial states $s^{(0)}$ and \emph{fair} schedules $\sigma$, the $(s^{(0)},\sigma)$-dynamics converges. We say that the system is \emph{r-convergent} if, for all initial states $s^{(0)}$ and \emph{r-fair} schedules $\sigma$, the $(s^{(0)},\sigma)$-dynamics converges.

\vspace{0.02in}\noindent {\bf Update messages.} Observe that, in our model, nodes' actions are \emph{immediately observable} to other nodes at the end of each time step (``\emph{perfect monitoring}''). While this is clearly unrealistic in some important real-life contexts (\emph{e.g.}, some of the environments considered below), this restriction only strengthens our main results, that are impossibility results.

\vspace{0.02in}\noindent {\bf Deterministic historyless dynamics.} Of special interest to us is the case that all reaction functions are deterministic and historyless. We observe that, in this case, stable states have a simple characterization. Each reaction function $f_i$ is deterministic and historyless and so can be specified by a function $g_i:A\rightarrow A_i$. Let $g=(g_1,\ldots,g_n)$. Observe that the set of all stable states (for all possible dynamics) is precisely the set of all fixed points of $g$. Below, when describing nodes' reaction functions that are deterministic and historyless we sometimes abuse notation and identify each $f_i$ with $g_i$ (treating $f_i$ as a function from $A$ to $A_i$). In addition, when all the reaction functions are also self-independent we occasionally treat each $f_i$ as a function from $A_{-i}$ to $A_i$.

\section{Non-Convergence Result}\label{sec:bounded}

We now present a general impossibility result for convergence of nodes' actions under bounded-recall dynamics in asynchronous, distributed computational environments.

\begin{theorem}\label{thm:historyless}
If each reaction function has bounded recall and is self-independent then the existence of multiple stable states implies that the system is not convergent.
\end{theorem}

We note that this result holds even if nodes' reaction functions are not stationary and are randomized (randomized initial states and activations are discussed in Section~\ref{sec:schedules}). We present the proof of Theorem~\ref{thm:historyless} in Appendix~\ref{ap:axiomatic}. We now discuss some aspects of our impossibility result.

\vspace{0.02in}\noindent {\bf Neither bounded recall nor self-independence alone implies non-convergence}
We show that the statement of Theorem~\ref{thm:historyless} does not hold if either the bounded-recall restriction, or the self-independence restriction, is removed.

\begin{example}\label{ex:not-consensus} {\bf(the bounded-recall restriction cannot be removed)}
There are two nodes, $1$ and $2$, each with the action space $\{x,y\}$. The deterministic and self-independent reaction functions of the nodes are as follows: node $2$ always chooses node $1$'s action; node $1$ will choose $y$ if node $2$'s action changed from $x$ to $y$ in the past, and $x$ otherwise. Observe that node $1$'s reaction function is not bounded-recall but can depend on the \emph{entire} history of interaction. We make the observations that the system is safe and has two stable states. Observe that if node $1$ chooses $y$ at some point in time due to the fact that node $2$'s action changed from $x$ to $y$, then it shall continue to do so thereafter; if, on the other hand, $1$ never does so, then, from some point in time onwards, node $1$'s action is constantly $x$. In both cases, node $2$ shall have the same action as node $1$ eventually, and thus convergence to one of the two stable states, $(x,x)$ and $(y,y)$, is guaranteed. Hence, two stable states exist and the system is convergent nonetheless
\end{example}

\begin{example}\label{ex:iindependent} {\bf(the self-independence restriction cannot be removed)}
There are two nodes, $1$ and $2$, each with action set $\{x,y\}$. Each node $i$'s a deterministic and historyless reaction function $f_i$ is as follows: $f_i(x,x)=y$; in all other cases the node always (re)selects its current action (\eg, $f_1(x,y)=x$, $f_2(x,y)=y$). Observe that the system has three stable states, namely all action profiles but $(x,x)$, yet can easily be seen to be convergent.
\end{example}

\noindent{\bf Connections to consensus protocols.} We now briefly discuss the interesting connections between Theorem~\ref{thm:historyless} and the non-termination result for failure-resilient consensus protocols in~\cite{FLP85}. We elaborate on this topic in Appendix~\ref{apx:connections}. Fischer \emph{et al.}~\cite{FLP85} explore when a group of processors can reach a consensus even in the presence of failures, and exhibit a breakthrough non-termination result. Our proof of Theorem~\ref{thm:historyless} uses a valency argument---an idea introduced in the proof of the non-termination result in~\cite{FLP85}.

Intuitively, the risk of protocol non-termination in~\cite{FLP85} stems from the possibility of failures; a computational node cannot tell whether another node is silent due to a failure or is simply taking a long time to react. We consider environments in which nodes/communication channels cannot fail, and so each node is guaranteed that all other nodes react after ``sufficiently long'' time. This guarantee makes reaching a consensus in the environment of~\cite{FLP85} easily achievable (see Appendix~\ref{apx:connections}). Unlike the results in~\cite{FLP85}, the possibility of nonconvergence in our framework stems from limitations on nodes' behaviors. Hence, there is no immediate translation from the result in~\cite{FLP85} to ours (and vice versa). To illustrate this point, we observe that in both Example~\ref{ex:not-consensus} and Example~\ref{ex:iindependent}, there exist two stable states and an initial state from which both stable states are reachable (a ``bivalent state''~\cite{FLP85}), yet the system is convergent (see Appendix~\ref{apx:connections}). This should be contrasted with the result in~\cite{FLP85} that establishes that the existence of an initial state from which two distinct outcomes are reachable implies the existence of a non-terminating execution.

We investigate the link between consensus protocols and our framework further in Appendix~\ref{ap:axiomatic}, where we take an axiomatic approach. We introduce a condition---``\emph{Independence of Decisions}'' (IoD)---that holds for both fault-resilient consensus protocols and for bounded-recall self-independent dynamics. We then factor the arguments in~\cite{FLP85} through IoD to establish a non-termination result that holds for both contexts, thus unifying the treatment of these dynamic computational environments.

\section{Games, Circuits, Networks, and Beyond}\label{sec:examples}

We present implications of our impossibility result in Section~\ref{sec:bounded} for several well-studied environments: game theory, circuit design, social networks and Internet protocols. We now briefly summarize these implications, that, we believe, are themselves of independent interest. See Appendix~\ref{apx:examples} for a detailed exposition
of the results in this section.

\vspace{0.02in}\noindent{\bf Game theory.} Our result, when cast into game-theoretic terminology, shows that if players' choices of strategies are not synchronized, then the existence of two (or more) pure Nash equilibria implies that a broad class of game dynamics (\emph{e.g.}, best-response dynamics with consistent tie-breaking) are not guaranteed to reach a pure Nash equilibrium. This result should be contrasted with positive results for such dynamics in the traditional synchronous game-theoretic environments.

\begin{theorem}
If there are two (or more) pure Nash equilibria in a game, then all bounded-recall self-independent dynamics can oscillate indefinitely for asynchronous player activations.
\end{theorem}
\begin{cor}
If there are two (or more) pure Nash equilibria in a game, then best-response dynamics, and bounded-recall best-response dynamics (studied in~\cite{Z08}), with consistent tie-breaking, can fail to converge to an equilibrium in asynchronous environments.
\end{cor}

\vspace{0.02in}\noindent{\bf Circuits.} Work on asynchronous circuits in computer architectures research explores the implications of asynchrony for circuit design~\cite{DN97}. We observe that a logic gate can be regarded as executing an inherently historyless reaction function that is independent of the gate's past and present ``state''. Thus, we show that our result has implications for the stabilization of asynchronous circuits.

\begin{theorem}
If two (or more) stable Boolean assignments exist for an asynchronous Boolean circuit, then that asynchronous circuit is not inherently stable.
\end{theorem}

\vspace{0.02in}\noindent{\bf Social networks.} Understanding the ways in which innovations, ideas, technologies, and practices, disseminate through social networks is fundamental to the social sciences. We consider the classic economic setting~\cite{Morris00} (that has lately also been approached by computer scientists~\cite{IKMW07}) where each decision maker has two technologies $\{A,B\}$ to choose from, and each node in the social network wishes to have the same technology as the majority of his ``friends'' (neighboring nodes in the social network). We exhibit a general impossibility result for this environment.

\begin{theorem}
In every social network, the diffusion of technologies can potentially never converge to a stable global state.
\end{theorem}

\vspace{0.02in}\noindent{\bf Networking.} We consider two basic networking
environments: (1) routing with the Border Gateway Protocol (BGP), that is the ``glue'' that holds together the smaller networks
that make up the Internet; and (2) the fundamental task of congestion control in communication networks, that is achieved
through a combination of mechanisms on \emph{end-hosts} (\eg, TCP), and on \emph{switches/routers} (\eg, RED and WFQ). We exhibit non-termination results for both these environments.

We abstract a recent result in~\cite{SSZ09} and prove that this result extends to several BGP-based \emph{multipath routing} protocols that have been proposed in the past few years.

\begin{theorem}~\cite{SSZ09}
If there are multiple stable routing trees in a network, then BGP is not safe on that network.
\end{theorem}

We consider the model for analyzing dynamics of congestion presented in~\cite{GSZS09}. We present the following result.

\begin{theorem}
If there are multiple capacity-allocation equilibria in the network then dynamics of congestion can oscillate indefinitely.
\end{theorem}

\section{$r$-Convergence and Randomness}~\label{sec:schedules}

We now consider the implications for convergence of two natural restrictions on schedules: $r$-fairness and randomization. See Appendix~\ref{apx:schedules} for a detailed exposition
of the results in this section.

\vspace{0.02in}\noindent{\bf Snakes in boxes and $r$-convergence.} Theorem~\ref{thm:historyless} deals with convergence and not $r$-convergence, and thus does not impose restrictions on the number of consecutive time steps in which a node can be nonactive. What happens if there is an upper bound on this number, $r$? We now show that if $r<n-1$ then sometimes convergence of historyless and self-independent dynamics is achievable even in the presence of multiple stable states (and so our impossibility result does not extend to this setting).

\begin{example} {\bf(a system that is convergent for $r<n-1$ but nonconvergent for $r=n-1$)}
There are $n\geq 2$ nodes, $1,\ldots,n$, each with the action space $\{x,y\}$. Nodes' deterministic, historyless and self-independent reaction functions are as follows. $\forall i\in [n]$, $f_i(x^{n-1})=x$ and $f_i$ always outputs $y$ otherwise. Observe that there exist two stable states: $x^n$ and $y^n$. Observe that if $r=n-1$ then the following oscillation is possible. Initially, only node $1$'s action is $y$ and all other nodes' actions are $x$. Then, nodes $1$ and $2$ are activated and, consequently, node $1$'s action becomes $x$ and node $2$'s action becomes $y$. Next, nodes $2$ and $3$ are activated, and thus $2$'s action becomes $x$ and $3$'s action becomes $y$. Then $3,4$ are activated, then $4,5$, and so on (traversing all nodes over and over again in cyclic order). This goes on indefinitely, never reaching one of the two stable states. Observe that, indeed, each node is activated at least once within every sequence of $n-1$ consecutive time steps. We observe however, that if $r < n-1$ then convergence is guaranteed. To see this, observe that if at some point in time there are at least two nodes whose action is $y$, then convergence to $y^n$ is guaranteed. Clearly, if all nodes' action is $x$ then convergence to $x^n$ is guaranteed. Thus, an oscillation is possible only if, in each time step, \emph{exactly} one node's action is $y$. Observe that, given our definition of nodes' reaction functions, this can only be if the activation sequence is (essentially) as described above, \ie, exactly two nodes are activated at a time. Observe also that this kind of activation sequence is impossible for $r<n-1$.

\end{example}

What about $r>n$? We use classical results in combinatorics regarding the size of a ``\emph{snake-in-the-box}'' in a hypercube~\cite{AK91} to construct systems are $r$-convergent for exponentially-large $r$'s, but are not convergent in general.

\begin{theorem}\label{thm:exp-r}
Let $n\in\N$ be sufficiently large. There exists a system $G$ with $n$ nodes, in which each node $i$ has two possible actions and each $f_i$ is deterministic, historyless and self-independent, such that $G$ is $r$-convergent for $r\in\Omega(2^{n})$, but $G$ is not $(r+1)$-convergent.
\end{theorem}

We note that the construction in the proof of Theorem~\ref{thm:exp-r} is such that there is a unique stable state. We believe that the same
ideas can be used to prove the same result for systems with multiple stable states but the exact way of doing this eludes us at the moment, and is left as an open question.

\begin{problem}
Prove that for every sufficiently large $n\in\N$, there exists a system $G$ with $n$ nodes, in which each node $i$ has two possible actions, each $f_i$ is deterministic, historyless and self-independent, and $G$ has multiple stable states, such that $G$ is $r$-convergent for $r\in\Omega(2^{n})$ but $G$ is not $(r+1)$-convergent.
\end{problem}

\noindent{\bf Does random choice (of initial state and schedule) help?} Theorem~\ref{thm:historyless} tells us that, for a broad class of dynamics, a system with multiple stable states is nonconvergent if the initial state and the node-activation schedule are chosen adversarially. Can we guarantee convergence if the initial state and schedule are chosen \emph{at random}?

\begin{example}\label{ex:random-schedule} {\bf(random choice of initial state and schedule might not help)}
There are $n$ nodes, $1,\ldots,n$, and each node has action space $\{x,y,z\}$. The (deterministic, historyless and self-independent) reaction function of each node $i\in\{3,\ldots,n\}$ is such that $f_i(x^{n-1})=x$; $f_i(z^{n-1})=z$; and $f_i=y$ for all other inputs. The (deterministic, historyless and self-independent) reaction function of each node $i\in\{1,2\}$ is such that $f_i(x^{n-1})=x$; $f_i(z^{n-1})=z$; $f_i(xy^{n-2})=y$; $f_i(y^{n-1})=x$; and $f_i=y$ for all other inputs. Observe that there are exactly two stable states: $x^n$ and $z^n$. Observe also that if nodes' actions in the initial state do not contain at least $n-1$ $x$'s, or at least $n-1$ $z$'s, then, from that moment forth, each activated node in the set $\{3,\ldots,n\}$ will choose the action $y$. Thus, eventually the actions of all nodes in $\{3,\ldots,n\}$ shall be $y$, and so none of the two stable states will be reached. Hence, there are $3^n$ possible initial states, such that only from $4n+2$ can a stable state be reached.  When choosing the initial state uniformly at random the probability of landing on a ``good'' initial state (in terms of convergence) is thus exponentially small.
\end{example}

\section{Complexity of Asynchronous Dynamics}\label{sec:complexity}

We now explore the communication complexity and computational complexity of determining
whether a system is convergent. We present hardness results in both models of computation even
for the case of deterministic and historyless adaptive heuristics. See Appendix~\ref{apx:complexity} for a detailed exposition
of the results in this section.

We first present the following communication complexity result whose proof
relies on combinatorial ``snake-in-the-box'' constructions~\cite{AK91}.

\begin{theorem}
Determining if a system with $n$ nodes, each with $2$ actions, is convergent requires
$\Omega(2^{n})$ bits. This holds even if all nodes have deterministic, historyless and self-independent
reaction functions.
\end{theorem}

The above communication complexity hardness result required the representation of the reaction functions to (potentially) be exponentially long.
What if the reaction functions can be succinctly described? We now present a strong computational complexity hardness result for the case that each reaction function $f_i$ is deterministic and historyless, and is given explicitly in the form of a boolean circuit (for each $a\in A$ the circuit outputs $f_i(a)$). We prove the following result.

\begin{theorem}\label{thm:pspace}
Determining if a system with $n$ nodes, each with a deterministic and historyless reaction function, is convergent is PSPACE-complete.
\end{theorem}

Our computational complexity result shows that even if nodes' reaction functions can be succinctly
represented, determining whether the system is convergent is PSPACE-complete. This
result, alongside its computational implications, implies that we cannot hope to have short
``witnesses'' of guaranteed asynchronous convergence (unless PSPACE $\subseteq$ NP). Proving the above PSPACE-completeness result for the case self-independent reaction functions seems challenging.

\begin{problem}
Prove that determining if a system with $n$ nodes, each with a deterministic self-independent and historyless reaction function, is convergent is PSPACE-complete.
\end{problem}

\section{Some Basic Observations Regarding No-Regret Dynamics}\label{sec:regret}

Regret minimization is fundamental to learning theory, and has strong connections to game-theoretic solution concepts; if each player in a repeated game executes a no-regret algorithm when selecting strategies, then convergence to an equilibrium is guaranteed in a variety of interesting contexts. The meaning of convergence, and the type of equilibrium reached, vary, and are dependent on the restrictions imposed on the game and on the notion of regret. Work on no-regret dynamics traditionally considers environments where all nodes are ``activated'' at each time step. We make the simple observation that, switching our attention to $r$-fair schedules (for every $r\in N_{+}$), if an algorithm has no regret in the classic setting, then it has no regret in this new setting as well (for all notions of regret). Hence, positive results from the regret-minimization literature extend to this asynchronous environment. See~\cite{BM07} for a thorough explanation about no-regret dynamics and see Appendix~\ref{apx:regret} for a detailed explanation about our observations. We now mention two implications of our observation and highlight two open problems regarding regret minimization.

\begin{obs}
When all players in a zero-sum game use no-external-regret algorithms then approaching or exceeding the minimax value of the game is guaranteed.
\end{obs}

\begin{obs}
When all players in a (general) game use no-swap-regret algorithms the empirical distribution of joint players' actions converges to a correlated equilibrium of the game.
\end{obs}

\begin{problem}
Give examples of repeated games for which there exists a schedule of player activations that is not $r$-fair for any $r\in N_{+}$ for which regret-minimizing dynamics do not converge to an equilibrium (for different notions of regret/convergence/equilibria).
\end{problem}

\begin{problem}
When is convergence of no-regret dynamics to an equilibrium guaranteed (for different notions of regret/convergence/equilibria) for all $r$-fair schedules for non-fixed $r$'s, that is, if when $r$ is a function of $t$?
\end{problem}

\section{Future Research}\label{sec:future}

In this paper, we have taken the first steps towards a complete
understanding of distributed computing with adaptive heuristics. We
proved a general non-convergence result and several hardness results
within this model, and also discussed some important aspects such as
the implications of fairness and randomness, as well as applications
to a variety of settings. We believe that we have but scratched the
surface in the exploration of the convergence properties of simple
dynamics in asynchronous computational environments, and many
important questions remain wide open. We now outline
several interesting directions for future research.

\vspace{0.02in}\noindent{\bf Other heuristics, convergence notions, equilibria.} We have considered specific adaptive heuristics, notions of convergence, and kinds of equilibria. Understanding the effects of asynchrony on other adaptive heuristics (\emph{e.g.},
better-response dynamics, fictitious play), for other notions of convergence (\emph{e.g.}, of the empirical distributions of play),
and for other kinds of equilibria (\emph{e.g.}, mixed Nash equilibria, correlated equilibria) is a broad and challenging
direction for future research.

\vspace{0.02in}\noindent{\bf Outdated and private information.}  We have not explicitly considered the effects of making decisions based on outdated information. We have also not dealt with the case that nodes' behaviors are dependent on private information, that is, the case that the dynamics are ``uncoupled''~\cite{HM03,HM06}.

\vspace{0.02in}\noindent{\bf Other notions of asynchrony.} We believe that better understanding the role
of degrees of fairness, randomness, and other restrictions on schedules from distributed computing literature, in achieving convergence to equilibrium points
is an interesting and important research direction.

\vspace{0.02in}\noindent{\bf Characterizing asynchronous convergence.} We still lack characterizations of asynchronous
convergence even for simple dynamics (\emph{e.g.}, deterministic and historyless).\footnote{Our
PSPACE-completeness result in Section~\ref{sec:complexity} eliminates
the possibility of short witnesses of guaranteed asynchronous
convergence unless PSPACE $\subseteq$ NP, but elegant
characterizations are still possible.}

\vspace{0.02in}\noindent{\bf Topological and knowledge-based approaches.} Topological~\cite{BG93,HS99,SZ00} and knowledge-based~\cite{HM90}
approaches have been very successful in addressing fundamental questions in distributed computing. Can these approaches shed new light on the implications of asynchrony for adaptive heuristics?

\vspace{0.02in}\noindent{\bf Further exploring the environments in Section~\ref{sec:examples}.} We have applied our non-convergence result to the environments described in Section~\ref{sec:examples}. These environments are of independent interest and are indeed the subject of extensive research. Hence, the further exploration of dynamics in these settings is important.

\section*{Acknowledgements}
We thank Danny Dolev, Alex Fabrikant, Idit Keidar, Jonathan Laserson, Nati Linial, Yishay
Mansour and Yoram Moses for helpful discussions. This work was initiated partly as a
result of the DIMACS Special Focus on Communication Security and
Information Privacy.

\appendix

\section{Connections to Consensus Protocols}\label{apx:connections}

There are interesting connections between our result and that of Fischer \emph{et al.}~\cite{FLP85} for \emph{fault-resilient consensus protocols}.~\cite{FLP85} studies the following environment: There is a group of \emph{processes}, each with an initial value in $\{0,1\}$, that communicate with each other via \emph{messages}. The objective is for all \emph{non-faulty} processes to eventually agree on some value $x\in\{0,1\}$, where the ``consensus'' $x$ must match the initial value of some process.~\cite{FLP85} establishes that no consensus protocol is resilient to even a single failure. One crucial ingredient for the proof of the result in~\cite{FLP85} is showing that there exists some initial configuration of  processes' initial values such that, from that configuration, the resulting consensus can be both $0$ and $1$ (the outcome depends on the specific ``schedule'' realized). Our proof of Theorem~\ref{thm:historyless} uses a valency argument---an idea introduced in the proof of the breakthrough non-termination result in~\cite{FLP85} for consensus protocols.

Intuitively, the risk of protocol nontermination in~\cite{FLP85} stems from the possibility of failures; a computational node cannot tell whether another node is silent due to a failure or is simply taking a long time to react. We consider environments in which nodes/communication channels do not fail. Thus, each node is guaranteed that after ``sufficiently many'' time steps all other nodes will react. Observe that in such an environment reaching a consensus is easy; one pre-specified node $i$ (the ``dictator'') waits until it learns all other nodes' inputs (this is guaranteed to happen as failures are impossible) and then selects a value $v_i$ and informs all other nodes; then, all other nodes select $v_i$. Unlike the results in~\cite{FLP85}, the possibility of nonconvergence in our framework stems from limitations on nodes' behaviors. We investigate the link between consensus protocols and our framework further in Appendix.~\ref{ap:axiomatic}, where we take an axiomatic approach. We introduce a condition---``\emph{Independence of Decisions}'' (IoD)---that holds for both fault-resilient consensus protocols and for bounded-recall self-independent dynamics. We then factor the arguments in~\cite{FLP85} through IoD to establish a non-termination result that holds for both contexts, thus unifying the treatment of these dynamic computational environments.

Hence, there is \emph{no} immediate translation from the result in~\cite{FLP85} to ours (and vice versa). To illustrate this point, let us revisit Example~\ref{ex:not-consensus}, in which the system is convergent, yet two stable states exist. We observe that in the example there is indeed an initial state from which both stable states are reachable (a ``bivalent state''~\cite{FLP85}). Consider the initial state $(y,x)$. Observe that if node $1$ is activated first (and alone), then it shall choose action $x$. Once node $2$ is activated it shall then also choose $x$, and the resulting stable state shall be $(x,x)$. However, if node $2$ is activated first (alone), then it shall choose action $y$. Once $1$ is activated it shall also choose action $y$, and the resulting stable state shall be $(y,y)$. Observe that in Example~\ref{ex:iindependent} too there exists an action profile $(x,x)$ from which multiple stable states are reachable yet the system is convergent.

\section{Games, Circuits, Networks, and Beyond}\label{apx:examples}

We present implications of our impossibility result in Section~\ref{sec:bounded} for several well-studied environments: game theory, circuit design, social networks and Internet protocols.

\subsection{Game Dynamics}\label{sec:examples-games}

\noindent{\bf The setting.} There are $n$ \emph{players}, $1,\ldots,n$. Each player $i$ has a \emph{strategy set} $S_i$. Let $S=\times_{j\in N}S_j$, and let $S_{-i}=\times_{j\in [n]\setminus\{i\}}S_j$. Each player $i$ has a \emph{utility function} $u_i:S\rightarrow S_i$. For each $s_i\in S_i$ and $s_{-i}\in S_{-i}$ let $(s_i,s_{-i})$ denote the strategy profile in which player $i$'s strategy is $s_i$ and all other players' strategies are as in $s_{-i}$. Informally, a \emph{pure Nash equilibrium} is a strategy profile from which no player wishes to unilaterally deviate.

\begin{definition} {\bf(pure Nash equilibria)}
We say that a strategy profile $\overline{s}=(\overline{s_1},\ldots,\overline{s_n})\in S$ is a \emph{pure Nash equilibrium} if, for each player $i$, $\overline{s_i}\in argmax_{s_i\in S_i}u_i(s_i,\overline{s_{-i}})$.
\end{definition}

One natural procedure for reaching a pure Nash equilibrium of a game is \emph{best-response dynamics}: the process starts at some arbitrary strategy profile, and players take turns ``best replying'' to other players' strategies until no player wishes to change his strategy. Convergence of best-response dynamics to pure Nash equilibria is the subject of extensive research in game theory and economics, and both positive~\cite{R73,MS96} and negative~\cite{HM03,HM06} results are known.

Traditionally, work in game theory on game dynamics (\emph{e.g.}, best-response dynamics) relies on the explicit or implicit premise that players' actions are somehow synchronized (in some contexts play is sequential, while in others it is simultaneous). We consider the realistic scenario that there is no computational center than can synchronize players' selection of strategies. We cast the above setting into the terminology of Section~\ref{sec:model} and exhibit an impossibility result for best-response, and more general, dynamics.

\vspace{0.05in}\noindent{\bf Computational nodes, action spaces.} The \emph{computational nodes} are the $n$ players. The \emph{action space} of each player $i$ is his strategy set $S_i$.

\vspace{0.05in}\noindent{\bf Reaction functions, dynamics.} Under best-response dynamics, each player constantly chooses a ``best response''
to the other players' most recent actions. Consider the case that players have consistent tie-breaking rules, \ie, the best response is always unique, and depends only on the others' strategies. Observe that, in this case, players' behaviors can be formulated as deterministic, historyless, and self-independent reaction functions. The dynamic interaction between players is as in Section~\ref{sec:model}.

\vspace{0.05in}\noindent{\bf Existence of multiple pure Nash equilibria implies non-convergence of best-response dynamics in asynchronous environments.} Theorem~\ref{thm:historyless} implies the following result:

\begin{theorem}
If there are two (or more) pure Nash equilibria in a game, then asynchronous best-response dynamics can potentially oscillate indefinitely.
\end{theorem}

In fact, Theorem~\ref{thm:historyless} implies that the above non-convergence result holds even for the broader class of randomized, bounded-recall and self-independent game dynamics, and thus also to game dynamics such as best-response with bounded recall and consistent tie-breaking rules (studied in~\cite{Z08}).

\subsection{Asynchronous Circuits}

\noindent{\bf The setting.} There is a Boolean circuit, represented as a directed graph $G$, in which vertices represent the circuit's inputs and the logic gates, and edges represent connections between the circuit's inputs and the logic gates and between logic gates. The activation of the logic gates is asynchronous. That is,
the gates' outputs are initialized in some arbitrary way, and then the update
of each gate's output, given its inputs, is uncoordinated and unsynchronized. We prove an impossibility
result for this setting, which has been extensively studied (see~\cite{DN97}).

\vspace{0.05in}\noindent{\bf Computational nodes, action spaces.} The computational nodes are
the inputs and the logic gates.  The \emph{action space} of
each node is $\{0,1\}$.

\vspace{0.05in}\noindent{\bf Reaction functions, dynamics.} Observe that each logic gate can be regarded as a function that only depends on its inputs' values. Hence, each logic gate can be modeled via a \emph{reaction function}. Interaction between the different circuit components
is as in Section~\ref{sec:model}.

\vspace{0.05in}\noindent{\bf Too much stability in circuits can lead to instability.}
\emph{Stable states} in this framework are assignments of Boolean values to the circuit inputs and the logic gates that are consistent with each gate's truth table (reaction function). We say that a Boolean circuit is \emph{inherently stable} if it is guaranteed to converge to a stable state regardless of the initial boolean assignment. The following theorem is derived from Theorem~\ref{thm:historyless}:

\begin{theorem}
If two (or more) stable Boolean assignments exist for an asynchronous Boolean circuit, then that asynchronous circuit is not inherently stable.
\end{theorem}

\subsection{Diffusion of Technologies in Social Networks}

\noindent{\bf The setting.} There is a social
network of users, represented by a directed graph in which users are the vertices and edges correspond to
friendship relationships. There are two competing technologies, $X$ and $Y$.
A user's utility from each technology depends on the number of that user's friends that
use that technology; the more friends use that technology the more desirable that technology is
to the user. That is, a user would always select the technology used by the majority of his
friends. We are interested in the dynamics of the diffusion of technologies. Observe that if,
initially, all users are using $X$, or all users are using $Y$, no user has an incentive to switch to a
different technology. Hence, there are always (at least) two distinct ``stable states'' (regardless
of the topology of the social network). Therefore, the terminology
of Section~\ref{sec:model} can be applied to this setting.

\vspace{0.05in}\noindent{\bf Computational nodes, actions spaces.} The users are the \emph{computational
nodes}. Each user $i$'s \emph{action space} consists of the two technologies $\{X,Y\}$.

\vspace{0.05in}\noindent{\bf Reaction functions, dynamics.} The reaction function of each user $i$ is defined as
follows: If at least half of $i$'s friends use technology $X$, $i$ selects technology $X$; otherwise, $i$ selects technology $Y$.
In our model of diffusion of technologies, users' choices of technology can be made simultaneously, as described in Section~\ref{sec:model}.

\vspace{0.05in}\noindent{\bf Instability of social networks.} Theorem~\ref{thm:historyless} implies the following:

\begin{theorem}
In every social network, the diffusion of technologies can potentially never converge to a stable global state.
\end{theorem}

\subsection{Interdomain Routing}

\noindent{\bf The setting.} The Internet is made up of smaller networks called \emph{Autonomous Systems} (ASes).
\emph{Interdomain routing} is the task of establishing routes between ASes, and is
handled by the \emph{Border Gateway Protocol} (BGP). In the standard model for analyzing BGP dynamics~\cite{GSW02}, there is a network
of \emph{source} ASes that wish to send traffic to a unique \emph{destination} AS $d$. Each AS $i$ has a \emph{ranking function} $<_i$ that specifies $i$'s strict preferences over all simple (loop-free) routes leading from $i$ to $d$.\footnote{ASes rankings of routes also reflect each AS's \emph{export policy} that specifies which routes that AS is willing to make available to each neighboring AS.} Under BGP, each AS constantly selects the ``best'' route that is available to it. See~\cite{GSW02} for more details. Guaranteeing \emph{BGP safety}, \ie, BGP convergence to a ``stable'' routing outcome is a fundamental desideratum that has been the subject of extensive work in both the networking and the standards communities. We now cast interdomain routing into the terminology of Section~\ref{sec:model}. We then obtain non-termination results for BGP and for proposals for new interdomain routing protocols (as corollaries of Theorem~\ref{thm:historyless}).

\vspace{0.05in}\noindent{\bf Computational nodes, action spaces.} The ASes are the \emph{computational nodes}. The \emph{action space} of each node $i$, $A_i$, is the set of all simple (loop-free) routes between $i$ and the destination $d$ that are exportable to $i$, and the empty route $\emptyset$.

\vspace{0.05in}\noindent{\bf Reaction functions, dynamics.} The \emph{reaction function} $f_i$ of node $i$ outputs, for every vector $\alpha$ containing routes to $d$ of all of $i$'s neighbors, a route $(i,j)R_j$ such that (1) $j$ is $i$'s neighbor; (2) $R_j$ is $j$'s route in $\alpha$; and (3) $R_j>_i R$ for all other routes $R$ in $\alpha$. If there is no such route $R_j$ in $\alpha$ then $f_i$ outputs $\emptyset$. Observe that the reaction function $f_i$ is deterministic, self-independent and historyless. The interaction between
nodes is as described in Section~\ref{sec:model}.

\vspace{0.05in}\noindent{\bf The multitude of stable routing trees implies global network instability.} Theorem~\ref{thm:historyless}
implies a recent result of Sami \etal.~\cite{SSZ09}, that shows that the existence of two (or more) stable routing trees to which BGP
can (potentially) converge implies that BGP is not safe. Importantly, the asynchronous model of Section~\ref{sec:model} is significantly \emph{more restrictive} than that of~\cite{SSZ09}. Hence, Theorem~\ref{thm:historyless} implies the non-termination result of Sami \etal.

\begin{theorem}~\cite{SSZ09}
If there are multiple stable routing trees in a network, then BGP is not safe on that network.
\end{theorem}

Over the past few years, there have been several proposals for BGP-based \emph{multipath routing} protocols, \ie, protocols that enable each node (AS) to send traffic along multiple routes, \emph{e.g.}, R-BGP~\cite{KKKM07} and Neighbor-Specific BGP~\cite{YSR09} (NS-BGP). Under both R-BGP and NS-BGP each computational node's actions are independent of its own past actions and are based on bounded recall of past interaction. Thus, Theorem~\ref{thm:historyless} implies the following:

\begin{theorem}
If there are multiple stable routing configurations in a network, then R-BGP is not safe on that network.
\end{theorem}

\begin{theorem}
If there are multiple stable routing configurations in a network, then NS-BGP is not safe on that network.
\end{theorem}

\subsection{Congestion Control}

\noindent{\bf The setting.} We now present the model of congestion control, studied in~\cite{GSZS09}.
There is a network of routers, represented by a directed graph $G=(V,E)$, where $|E|\geq 2$, in which vertices represent routers, and edges represent communication links. Each edge has capacity $c_e$. There are $n$ \emph{source-target pairs} of vertices $(s_i,t_i)$, termed ``\emph{connections}'', that represent communicating pairs of end-hosts. Each source-target pair $(s_i,t_i)$
is connected via some \emph{fixed} route, $R_i$. Each source $s_i$ transmits at a \emph{constant} rate
$\gamma_i>0$.\footnote{This is modeled via the addition of an edge $e=(u,s_i)$ to $G$, such that $c_e=\gamma_i$,
and $u$ has no incoming edges.} Routers have \emph{queue management}, or \emph{queueing}, policies, that dictate
how traffic traversing a router's outgoing edge should be divided between the connections whose routes traverse that edge.
The network is asynchronous and so routers' queueing decisions can be made simultaneously. See~\cite{GSZS09} for more details.

\vspace{0.05in}\noindent{\bf Computational nodes, action spaces} The \emph{computational nodes} are the edges. The \emph{action space} of each edge $e$ intuitively consists of all possible way to divide traffic going through $e$ between the connections whose routes
traverse $e$. More formally, for every edge $e$, let $N(e)$ be the number connections whose paths go through $e$.
$e$'s action space is then $A_i=\{x=(x_1,\ldots,x_{N(e)})|x_i\in \R_{\geq 0}^{N(e)}\ \mathrm{and}\ \Sigma_i x_i\leq c_e\}$.

\vspace{0.05in}\noindent{\bf Reaction functions, dynamics.} Each edge $e$'s \emph{reaction function}, $f_e$, models
the queueing policy according to which $e$'s capacity is shared: for every $N(e)$-tuple of nonnegative incoming flows
$(w_1,w_2,\ldots, w_{N(e)})$, $f_e$ outputs an action
$(x_1,\ldots,x_{N(e)})\in A_i$ such that $\forall i \in [N(e)]\  w_i \geq x_i$
(a connection's flow leaving the edge cannot be bigger than that
connection's flow entering the edge).
The interaction between the edges is as described in Section~\ref{sec:model}.

\vspace{0.05in}\noindent{\bf Multiple equilibria imply potential fluctuations of connections' throughputs.} \cite{GSZS09} shows that,
while one might expect that if sources transmit flow at a constant rate, flow will also be \emph{received} at a constant rate, this is
not necessarily the case. Indeed,~\cite{GSZS09} presents examples in which connections' throughputs can potentially fluctuate ad infinitum.
Equilibria (which correspond to stable states in Section~\ref{sec:model}), are global configurations of connections' flows on edges such that connections' incoming and outgoing flows on each edge are consistent with the queue management policy of the router controlling that edge. Using Theorem~\ref{thm:historyless}, we can obtain the following impossibility result:

\begin{theorem}
If there are multiple capacity-allocation equilibria in the network then dynamics of congestion can potentially oscillate indefinitely.
\end{theorem}

\section{$r$-Convergence and Randomness}~\label{apx:schedules}

We now consider the implications for convergence of two natural restrictions on schedules: $r$-fairness and randomization.

\subsection{\bf Snakes in Boxes and r-Convergence.} Theorem~\ref{thm:historyless} deals with convergence and not $r$-convergence, and thus does not impose restrictions on the number of consecutive time steps in which a node can be nonactive. What happens if there is an upper bound on this number, $r$? We now show that if $r<n-1$ then sometimes convergence of historyless and self-independent dynamics is achievable even in the presence of multiple stable states (and so our impossibility result breaks).

\begin{example} {\bf(a system that is convergent for $r<n-1$ but nonconvergent for $r=n-1$)}
There are $n\geq 2$ nodes, $1,\ldots,n$, each with the action space $\{x,y\}$. Nodes' deterministic, historyless and self-independent reaction functions are as follows. $\forall i\in [n]$, $f_i(x^{n-1})=x$ and $f_i$ always outputs $y$ otherwise. Observe that there exist two stable states: $x^n$ and $y^n$. Observe that if $r=n-1$ then the following oscillation is possible. Initially, only node $1$'s action is $y$ and all other nodes' actions are $x$. Then, nodes $1$ and $2$ are activated and, consequently, node $1$'s action becomes $x$ and node $2$'s action becomes $y$. Next, nodes $2$ and $3$ are activated, and thus $2$'s action becomes $x$ and $3$'s action becomes $y$. Then $3,4$ are activated, then $4,5$, and so on (traversing all nodes over and over again in cyclic order). This goes on indefinitely, never reaching one of the two stable states. Observe that, indeed, each node is activated at least once within every sequence of $n-1$ consecutive time steps. We observe however, that if $r < n-1$ then convergence is guaranteed. To see this, observe that if at some point in time there are at least two nodes whose action is $y$, then convergence to $y^n$ is guaranteed. Clearly, if all nodes' action is $x$ then convergence to $x^n$ is guaranteed. Thus, an oscillation is possible only if, in each time step, \emph{exactly} one node's action is $y$. Observe that, given our definition of nodes' reaction functions, this can only be if the activation sequence is (essentially) as described above, \ie, exactly two nodes are activated at a time. Observe also that this kind of activation sequence is impossible for $r<n-1$.

\end{example}

What about $r>n$? We use classical results in combinatorics regarding the size of a ``\emph{snake-in-the-box}'' in a hypercube~\cite{AK91} to show that some systems are $r$-convergent for exponentially-large $r$'s, but are not convergent in general.

\newtheorem{thm-exp-r}{Theorem~\ref{thm:exp-r}}
\begin{thm-exp-r}
Let $n\in\N$ be sufficiently large. There exists a system $G$ with $n$ nodes, in which each node $i$ has two possible actions and each $f_i$ is deterministic, historyless and self-independent, such that

\begin{enumerate}

\item $G$ is $r$-convergent for $r\in\Omega(2^{n})$;

\item $G$ is not $(r+1)$-convergent.

\end{enumerate}
\end{thm-exp-r}
\begin{proof}
Let the action space of each of the $n$ nodes be $\{x,y\}$. Consider the possible action profiles of nodes $3,\ldots,n$, \ie, the set $\{x,y\}^{n-2}$. Observe that this set of actions can be regarded as the $(n-2)$-hypercube $Q_{n-2}$, and thus can be visualized as the graph whose vertices are indexed by the binary $(n-2)$-tuples and such that two vertices are adjacent iff the corresponding $(n-2)$-tuples differ in exactly one coordinate.

\begin{definition} {\bf(chordless paths, snakes)}
A \emph{chordless path} in a hypercube $Q_n$ is a path $P=(v_0,\ldots,v_w)$ such that for each $v_i,v_j$ on $P$, if $v_i$ and $v_j$ are neighbors in $Q_n$ then $v_j\in \{v_{i-1},v_{i+1}\}$. A \emph{snake} in a hypercube is a simple chordless cycle.
\end{definition}

The following result is due to Abbot and Katchalski~\cite{AK91}.

\begin{theorem}~\cite{AK91}
Let $t\in\N$ be sufficiently large. Then, the size $|S|$ of a maximal snake in the $z$-hypercube $Q_z$ is at least $\lambda\times 2^z$ for some $\lambda\geq 0.3$.
\end{theorem}

Hence, the size of a maximal snake in the $Q_{n-2}$ hypercube is $\Omega(2^{n})$. Let $S$ be a maximal snake in $\{x,y\}^{n-2}$. W.l.o.g we can assume that $x^{n-2}$ is on $S$ (otherwise we can rename nodes' actions so as to achieve this). Nodes deterministic, historyless and self-independent are as follows:

\begin{itemize}

\item Node $i\in\{1,2\}$: $f_i(x^{n-1})=x$; $f_i=y$ otherwise.

\item Node $i\in\{3,\ldots,n\}$: if the actions of nodes $1$ and $2$ are both $y$ then the action $y$ is chosen, \ie, $f_i(yy*\ldots*)=y$; otherwise, $f_i$ only depends on the actions of nodes in $\{3,\ldots,n\}$ and therefore to describe $f_i$ it suffices to orient the edges of the hypercube $Q_{n-2}$ (an edge from one vertex to another vertex that differs from it in the $i$th coordinate determines the outcome of $f_i$ for both). This is done as follows: orient the edges in $S$ so as to create a cycle (in one of two possible ways); orient edges between vertices not in $S$ to vertices in $S$ towards the vertices in $S$; orient all other edges arbitrarily.

\end{itemize}

\begin{obs}
$x^n$ is the unique stable state of the system.
\end{obs}

\begin{obs}
If, at some point in time, both nodes $1$ and $2$'s actions are $y$ then convergence to the $y^n$ stable state is guaranteed.
\end{obs}

\begin{claim}\label{obs:must}
If there is an oscillation then there must be infinitely many time steps in which the actions of nodes $2,\ldots,n$ are $x^{n-1}$.
\end{claim}
\begin{proof}
Consider the case that the statement does not hold. In that case, from some moment forth, node $1$ never sees the actions $x^{n-1}$
and so will constantly select the action $y$. Once that happens, node $2$ shall also not see the actions $x^{n-1}$ and will thereafter also select $y$. Convergence to $y^n$ is then guaranteed.
\end{proof}

We now show that the system is convergent for $r<|S|$, but is nonconvergent if $r=|S|$. The theorem follows.

\begin{claim}
If $r<|S|$ then convergence to the stable state $y^n$ is guaranteed.
\end{claim}
\begin{proof}
Observation~\ref{obs:must} establishes that in an oscillation there must be infinitely many time steps in which the actions of nodes $2,\ldots,n$ are $x^{n-1}$. Consider one such moment in time. Observe that in the subsequent time steps nodes' action profiles will inevitably change as in $S$ (given our definition of nodes' $3,\ldots,n$ reaction functions). Thus, once the action profile is no longer $x^{n-1}$ there are at least $|S|-1$ time steps until it goes back to being $x^{n-1}$. Observe that if $1$ and $2$ are activated at some point in the intermediate time steps (which is guaranteed as $r<|S|$) then the actions of both shall be $y$ and so convergence to $y^n$ is guaranteed.
\end{proof}

\begin{claim}
If $r=|S|$ then an oscillation is possible.
\end{claim}
\begin{proof}
The oscillation is as follows. Start at $x^n$ and activate both $1$ and $2$ (this will not change the action profile). In the $|S|-1$ subsequent time steps activate all nodes but $1$ and $2$ until $x^n$ is reached again. Repeat ad infinitum.
\end{proof}

\end{proof}

We note that the construction in the proof of Theorem~\ref{thm:exp-r} is such that there is a unique stable state. We believe that the same
ideas can be used to prove the same result for systems with multiple stable states but the exact way of doing this eludes us at the moment, and is left as an open question.

\begin{problem}
Prove that for every sufficiently large $n\in\N$, there exists a system $G$ with $n$ nodes, in which each node $i$ has two possible actions and each $f_i$ is deterministic, historyless and self-independent, such that

\begin{enumerate}

\item $G$ is $r$-convergent for $r\in\Omega(2^{n})$;

\item $G$ is not $(r+1)$-convergent;

\item There are multiple stable states in $G$.

\end{enumerate}

\end{problem}

\subsection{Does Random Choice (of Initial State and Schedule) Help?}

Theorem~\ref{thm:historyless} tells us that a system with multiple stable states is nonconvergent if the initial state and the node-activation schedule are chosen adversarially. Can we guarantee convergence if the initial state and schedule are chosen \emph{at random}?

\begin{example}\label{ex:apx-random-schedule} {\bf(random choice of initial state and schedule might not help)}
There are $n$ nodes, $1,\ldots,n$, and each node has action space $\{x,y,z\}$. The (deterministic, historyless and self-independent) reaction function of each node $i\in\{3,\ldots,n\}$ is such that $f_i(x^{n-1})=x$; $f_i(z^{n-1})=z$; and $f_i=y$ for all other inputs. The (deterministic, historyless and self-independent) reaction function of each node $i\in\{1,2\}$ is such that $f_i(x^{n-1})=x$; $f_i(z^{n-1})=z$; $f_i(xy^{n-2})=y$; $f_i(y^{n-1})=x$; and $f_i=y$ for all other inputs. Observe that there are exactly two stable states: $x^n$ and $z^n$. Observe also that if nodes' actions in the initial state do not contain at least $n-1$ $x$'s, or at least $n-1$ $z$'s, then, from that moment forth, each activated node in the set $\{3,\ldots,n\}$ will choose the action $y$. Thus, eventually the actions of all nodes in $\{3,\ldots,n\}$ shall be $y$, and so none of the two stable states will be reached. Hence, there are $3^n$ possible initial states, such that only from $4n+2$ can a stable state be reached.
\end{example}

Example~\ref{ex:apx-random-schedule} presents a system with multiple stable states such that from most initial states \emph{all} possible choices of schedules do not result in a stable state. Hence, when choosing the initial state uniformly at random the probability of landing on a ``good'' initial state (in terms of convergence) is exponentially small.

\section{Complexity of Asynchronous Dynamics}\label{apx:complexity}

We now explore the communication complexity and computational complexity of determining
whether a system is convergent. We present hardness results in both models of computation even
for the case of deterministic and historyless adaptive heuristics. Our
computational complexity result shows that even if nodes' reaction functions can be succinctly
represented, determining whether the system is convergent is PSPACE-complete. This intractability
result, alongside its computational implications, implies that we cannot hope to have short
``witnesses'' of guaranteed asynchronous convergence (unless PSPACE $\subseteq$ NP).

\subsection{Communication Complexity}

We prove the following communication complexity result, that shows that, in general, determining
whether a system is convergent cannot be done efficiently.

\begin{theorem}
Determining if a system with $n$ nodes, each with $2$ actions, is convergent requires
$\Omega(2^{n})$ bits. This holds even if all nodes have deterministic, historyless and self-independent
reaction functions.
\end{theorem}
\begin{proof}
To prove our result we present a reduction from the following well-known problem in communication complexity theory.

\vspace{0.05in}\noindent 2-party SET DISJOINTNESS: There are two parties, Alice and Bob. Each party holds a subset of $\{1,\ldots,q\}$; Alice
holds the subset $E^A$ and Bob holds the subset $E^B$. The objective is to determine whether $E^A\cap E^B=\emptyset$. The following is well known.

\begin{theorem}
Determining whether $E^A\cap E^B=\emptyset$ requires (in the worst case) the communication of $\Omega(q)$ bits. This lower bound applies to randomized protocols with bounded $2$-sided error and also to nondeterministic protocols.
\end{theorem}

We now present a reduction
from 2-party SET DISJOINTNESS to the question of determining whether a system with deterministic, historyless and self-independent reaction functions is convergent. Given an instance of SET-DISJOINTNESS we construct a system with $n$ nodes, each with two actions, as follows (the relation between the parameter $q$ in SET DISJOINTNESS and the number of nodes $n$ is to be specified later). Let the action space of each node be $\{x,y\}$. We now define the reaction functions of the nodes. Consider the possible action profiles of nodes $3,\ldots,n$, \ie, the set $\{x,y\}^{n-2}$. Observe that this set of actions can be regarded as the $(n-2)$-hypercube $Q_{n-2}$, and thus can be visualized as the graph whose vertices are indexed by the binary $(n-2)$-tuples and such that two vertices are adjacent if and only if the corresponding $(n-2)$-tuples differ in exactly one coordinate.

\begin{definition} {\bf(chordless paths, snakes)}
A \emph{chordless path} in a hypercube $Q_n$ is a path $P=(v_0,\ldots,v_w)$ such that for each $v_i,v_j$ on $P$, if $v_i$ and $v_j$ are neighbors in $Q_n$ then $v_j\in \{v_{i-1},v_{i+1}\}$. A \emph{snake} in a hypercube is a simple chordless cycle.
\end{definition}

The following result is due to Abbot and Katchalski~\cite{AK91}.

\begin{theorem}~\cite{AK91}
Let $t\in\N$ be sufficiently large. Then, the size $|S|$ of a maximal snake in the $z$-hypercube $Q_z$ is at least $\lambda\times 2^z$ for some $\lambda\geq 0.3$.
\end{theorem}

Hence, the size of a maximal snake in the $Q_{n-2}$ hypercube is $\Omega(2^{n})$. Let $S$ be a maximal snake in $\{x,y\}^{n-2}$. We now show our reduction from SET DISJOINTNESS with $q=|S|$. We identify each element $j\in \{1\ldots,q\}$ with a unique vertex $v_j\in S$. W.l.o.g we can assume that $x^{n-2}$ is on $S$ (otherwise we can rename nodes' actions to achieve this). For ease of exposition we also assume that $y^{n-2}$ is not on $S$ (getting rid of this assumption is easy). Nodes' reaction functions are as follows.

\begin{itemize}

\item Node $1$: If $v_j=(v_{j,1},\ldots,v_{j,n-2})\in S$ is a vertex that corresponds to an element $j\in E^A$, then $f_1(y,v_{j,1},\ldots,v_{j,n-2})=x$; otherwise, $f_1$ outputs $y$.

\item Node $2$: If $v_j=(v_{j,1},\ldots,v_{j,n-2})\in S$ is a vertex that corresponds to an element $j\in E^B$, then $f_2(y,v_{j,1},\ldots,v_{j,n-2})=x$; otherwise, $f_2$ outputs $y$.

\item Node $i\in\{3,\ldots,n\}$: if the actions of nodes $1$ and $2$ are \emph{not} both $x$ then the action $y$ is chosen; otherwise, $f_i$ only depends on actions of nodes in $\{3,\ldots,n\}$ and therefore to describe $f_i$ it suffices to orient the edges of the hypercube $Q_{n-2}$ (an edge from one vertex to another vertex that differs from it in the $i$th coordinate determines the outcome of $f_i$ for both). This is done as follows: orient the edges in $S$ so as to create a cycle (in one of two possible ways); orient edges between vertices not in $S$ to vertices in $S$ towards the vertices in $S$; orient all other edges arbitrarily.
\end{itemize}

\begin{obs}
$y^n$ is the unique stable state of the system.
\end{obs}

In our reduction Alice simulates node $1$ (whose reaction function is based on $E^A$), Bob simulates node $1$ (whose reaction function is based on $E^B$), and one of the two parties simulates all other nodes (whose reaction functions are not based on neither $E^A$ nor $E^B$). The theorem now follows from the combination of the following claims:

\begin{claim}\label{claim:infinite}
In an oscillation it must be that there are infinitely many time steps in which both node $1$ and $2$'s actions are $x$.
\end{claim}
\begin{proof}
By contradiction. Consider the case that from some moment forth it is never the case that both node $1$ and $2$'s actions are $x$. Observe that
from that time onwards the nodes $3,\ldots,n$ will always choose the action $y$. Hence, after some time has passed the actions
of all nodes in $\{3,\ldots,n\}$ will be $y$. Observe that whenever nodes $1$ and $2$ are activated thereafter they shall choose the action $y$ and so we have convergence to the stable state $y^n$.
\end{proof}

\begin{claim}
The system is not convergent iff $E^A\cap E^B\neq\emptyset$.
\end{claim}
\begin{proof}
We know (Claim~\ref{claim:infinite}) that if there is an oscillation then there are infinitely many time steps in which both node $1$ and $2$'s actions are $x$. We argue that this implies that there must be infinitely many time steps in which both nodes select action $x$ \emph{simultaneously}. Indeed, recall that node $1$ only chooses action $x$ if node $2$'s action is $y$, and vice versa, and so if both nodes never choose $x$ simultaneously, then it is never the case that both nodes' actions are $x$ at the same time step (a contradiction). Now, when is it possible for both
$1$ and $2$ to choose $x$ at the same time? Observe that this can only be if the actions of the nodes in $\{3,\ldots,n\}$ constitute an element that is in both $E^A$ and $E^B$. Hence, $E^A\cap E^B\neq\emptyset$.
\end{proof}

\end{proof}

\subsection{Computational Complexity}

The above communication complexity hardness result required the representation of the reaction functions to (potentially) be exponentially long.
What if the reaction functions can be succinctly described? We now present a strong computational complexity hardness result for the case that each reaction function $f_i$ is deterministic and historyless, and is given explicitly in the form of a boolean circuit (for each $a\in A$ the circuit outputs $f_i(a)$).

\newtheorem{thm-pspace}{Theorem~\ref{thm:pspace}}
\begin{thm-pspace}
Determining if a system with $n$ nodes, each with a deterministic and historyless reaction function, is convergent is PSPACE-complete.
\end{thm-pspace}
\begin{proof}
Our proof is based on the proof of Fabrikant and Papadimitriou~\cite{FP08} that BGP safety is PSPACE-complete. Importantly, the result in~\cite{FP08} does not imply Theorem~\ref{thm:pspace} since~\cite{FP08} only considers dynamics in which nodes are activated one at a time. We present a reduction from the following problem.

\vspace{0.05in}\noindent STRING NONTERMINATION: The input is a function $g:\Gamma^t\rightarrow \Gamma\cup\{halt\}$, for some alphabet $\Gamma$, given in the form of a boolean circuit. The objective is to determine whether there exists an initial string $T=(T_0,\ldots, T_{t-1})\in\Gamma^t$ such that the following procedure \emph{does not} halt.

\begin{enumerate}

\item $i$:=0

\item While $g(T)\neq halt$ do

\begin{itemize}

\item $T_i:=g(T)$

\item $i:=(i+1)\ modulu\ t$

\end{itemize}

\end{enumerate}

\noindent STRING NONTERMINATION is closely related to STRING HALTING from~\cite{FP08} and is also PSPACE-complete. We now present a reduction
from STRING NONTERMINATION to the question of determining whether a system with deterministic and historyless reaction functions is convergent.

We construct a system with $n=t+1$ nodes. The node set is divided into $t$ ``\emph{index nodes}'' $0,\ldots,t-1$ and a single ``\emph{counter node}'' $x$. The action space of each index node is $\Gamma\cup\{halt\}$ and the action space of the counter node is $\{0,\ldots,t-1\}\times(\Gamma\cup\{halt\})$. Let $a=(a_0,\ldots,a_{t-1},a_x)$ be an action profile of the nodes, where $a_x=(j,\gamma)$ is the action of the counter node. We now define the deterministic and historyless reaction functions of the nodes:

\begin{itemize}

\item The reaction function of index node $i\in\{0,\ldots,t-1\}$, $f_i$: if $\gamma=halt$, then $f_i(a)=halt$; otherwise, if $j=i$, and $a_j\neq\gamma$, then $f_i(a)=\gamma$; otherwise, $f_i(a)=a_i$.

\item The reaction function of the counter node, $f_x$: if $\gamma=halt$, then $f_x(a)=a_x$; if $a_j=\gamma$, then $f_i(a)=((j+1)\ modulu\ t,g(a_0,\ldots,a_{t-1})$; otherwise $f_i(a)=a_x$.

\end{itemize}

The theorem now follows from the following claims that, in turn, follow from our construction:

\begin{claim}
$(halt,\ldots,halt)$ is the unique stable state of the system.
\end{claim}
\begin{proof}
Observe that $(halt,\ldots,halt)$ is indeed a stable state of the system. The uniqueness of this stable state is proven via a simple case-by-case analysis.
\end{proof}

\begin{claim}
If there exists an initial string $T=(T_0,\ldots,T_{t-1}$) for which the procedure does not terminate then there exists an initial state from which
the system does not converge to the stable state $(halt,\ldots,halt)$ regardless of the schedule chosen.
\end{claim}
\begin{proof}
Consider the evolution of the system from the initial state in which the action of index node $i$ is $T_i$ and the action of the counter node is $(0,g(T))$.
\end{proof}

\begin{claim}
If there does not exist an initial string T for which the procedure does not terminate then the system is convergent.
\end{claim}
\begin{proof}
Observe that if there is an initial state $a=(a_0,\ldots,a_{t-1},a_x)$ and a fair schedule for which the system does not converge to the unique
stable state then the procedure does not halt for the initial string $T=(a_0,\ldots,a_{t-1})$.
\end{proof}
\end{proof}

Proving the above PSPACE-completeness result for the case self-independent reaction functions seems challenging.

\begin{problem}
Prove that determining if a system with $n$ nodes, each with a deterministic self-independent and historyless reaction function, is convergent is PSPACE-complete.
\end{problem}

\section{Some Basic Observations Regarding No-Regret Dynamics}\label{apx:regret}

Regret minimization is fundamental to learning theory. The basic setting is as follows. There is a space of $m$ actions (\emph{e.g.}, possible routes to work), which we identify with the set $[m]=\{1,\ldots,m\}$. In each time step $t\in\{1,\ldots\}$, an adversary selects a profit function $p_t:[m]\rightarrow [0,1]$ (\emph{e.g.}, how fast traffic is flowing along each route), and the (randomized) algorithm chooses a distribution $D_t$ over the elements in $[m]$. When choosing $D_t$ the algorithm can only base its decision on the profit functions $p_1,\ldots,p_{t-1}$, and not on $p_t$ (that is revealed only after the algorithm makes its decision). The algorithm's gain at time $t$ is $g_t=\Sigma_{j\in [m]}\ D_t(j)p_t(j)$, and its accumulated gain at time $t$ is $\Sigma_{i=1}^t\ g_t$ . Regret analysis is useful for designing adaptive algorithms that fair well in such uncertain environments. The motivation behind regret analysis is ensuring that, over time, the algorithm performs at least as well in retrospect as some alternative ``simple'' algorithm.

We now informally present the three main notions of regret (see~\cite{BM07} for a thorough explanation): (1) \emph{External regret} compares the algorithm's performance to that of simple algorithms that select the exact same action in each time step (\emph{e.g.}, ``you should have always taken Broadway, and never chosen other routes''). (2) \emph{Internal regret} and \emph{swap regret} analysis compares
the gain from the sequence of actions actually chosen to that derived from replacing every occurrence of an action $i$ with another action $j$ (\emph{e.g.}, ``every time you chose Broadway you should have taken $7$th Avenue instead). While internal regret analysis allows only \emph{one} action to be replaced by another, swap regret analysis considers all mappings from $[m]$ to $[m]$. The algorithm has \emph{no (external/internal/swap) regret} if the gap between the algorithm's gain and the gain from the best alternative policy allowed vanishes with time.

Regret minimization has strong connections to game-theoretic solution concepts. If each player in a repeated game executes a no-regret algorithm when selecting strategies, then convergence to an equilibrium is guaranteed in a variety of interesting contexts. The notion of convergence, and the kind of equilibrium reached, vary, and are dependent on the restrictions imposed on the game and on the type of regret being minimized (\emph{e.g.}, in zero-sum games, no-external-regret algorithms are guaranteed to approach or exceed the minimax value of the game; in general games, if all players minimize swap regret, then the empirical distribution of joint players' actions converges to a correlated equilibrium, \emph{etc.}). (See~\cite{BM07} and references therein). Importantly, these results are all proven within a model of interaction in which each player selects a strategy in each and every time step.

We make the following simple observation. Consider a model in which the adversary not only chooses the profit functions
but also has the power not to allow the algorithm to select a new distribution over actions in some time steps. That is, the adversary also selects a schedule $\sigma$ such that $\forall t\in\N_{+}$, $\sigma(t)\in\{0,1\}$, where $0$ and $1$ indicate whether the algorithm is not activated, or activated, respectively. We restrict the schedule to be $r$-fair, in the sense that the schedule chosen must be such that the algorithm is activated at least once in every $r$ consecutive time steps. If the algorithm is activated at time $t$ and not activated again until time $t+\beta$ then it holds that $\forall s\in\{t+1,\ldots,t+\beta-1\}$, $D_s=D_t$ (the algorithm cannot change its probability distribution over actions while not activated). We observe that if an algorithm has no regret in the above setting (for all three notions of regret), then it has no regret in this setting as well. To see this, simply observe that if we regard each batch of time steps in which the algorithms is not activated as one ``meta time step'', then this new setting is equivalent to the traditional setting (with $p_t:[m]\rightarrow [0,r]$ for all $t\in\N_{+}$).

This observation, while simple, is not uninteresting, as it implies that \emph{all} regret-based results for repeated games continue to hold even if players' order of activation is \emph{asynchronous} (see Section~\ref{sec:model} for a formal exposition of asynchronous interaction), so long as the schedule of player activations is $r$-fair for some $r\in N_{+}$. We mention two implications of this observation.

\begin{obs}
When all players in a zero-sum game use no-external-regret algorithms then approaching or exceeding the minimax value of the game is guaranteed.
\end{obs}

\begin{obs}
When all players in a (general) game use no-swap-regret algorithms the empirical distribution of joint players' actions converges to a correlated equilibrium of the game.
\end{obs}

\begin{problem}
Give examples of repeated games for which there exists a schedule of player activations that is not $r$-fair for any $r\in N_{+}$ for which regret-minimizing dynamics do not converge to an equilibrium (for different notions of regret/convergence/equilibria).
\end{problem}

\begin{problem}
When is convergence of no-regret dynamics to an equilibrium guaranteed (for different notions of regret/convergence/equilibria) for all $r$-fair schedules for non-fixed $r$'s, that is, if when $r$ is a function of $t$?
\end{problem}

\section{An Axiomatic Approach}\label{ap:axiomatic}

We now use (a slight variation of) the framework of Taubenfeld, which he used to study resilient consensus protocols~\cite{taubenfeld91}, to prove Thm.~\ref{thm:historyless}.  We first (Sec.~\ref{ssec:erp}) define \emph{runs} as sequences of \emph{events}; unlike Taubenfeld, we allow infinite runs.  A protocol is then a collection of runs (which must satisfy some natural conditions like closure under taking prefixes).  We then define colorings of runs (which correspond to outcomes that can be reached by extending a run in various ways) and define the \iod\ property.

The proof of Thm.~\ref{thm:historyless} proceeds in two steps.  First, we show that \emph{any} protocol that satisfies \iod\ has some (fair, as formalized below), non-terminating activation sequence.  We then show that protocols that satisfy the hypotheses of Thm.~\ref{thm:historyless} also satisfy \iod.

\subsection{Proof Sketch}

\begin{proof}[Proof Sketch]
The proof follows the axiomatic approach of Taubenfeld~\cite{taubenfeld91} in defining asynchronous protocols in which states are colored by sets of colors; the set of colors assigned to a state must be a superset of the set of colors assigned to any state that is reachable (in the protocol) from it.  We then show that any such protocol that satisfies a certain pair of properties (which we call \emph{Independence of Decisions} or \emph{IoD}) and that has a polychromatic state must have a non-terminating fair run in which all states are polychromatic.

For protocols with $1$-recall, self-independence, and stationarity, we consider (in order to reach a contradiction) protocols that are guaranteed to converge.  Each starting state is thus guaranteed to reach only stable states; we then color each state according to the outcomes that are reachable from that state.   We show that, under this coloring, such protocols satisfy IoD and that, as in consensus protocols, the existence of multiple stable states implies the existence of a polychromatic state.  The non-terminating, polychromatic, fair run that is guaranteed to exist is, in the context, exactly the non-convergent protocol run claimed by the theorem statement.  We then show that this may be extended to non-stationary protocols with $k$-recall (for $k > 1$).
\end{proof}

\subsection{Events, Runs, and Protocols}\label{ssec:erp}

Events are the atomic actions that are used to build runs of a protocol.  Each event is associated with one or more principals; these should be thought of as the principals who might be affected by the event (\eg, as sender or receiver of a message), with the other principals unable to see the event.  We start with the following definition.
\begin{definition}[Events and runs]
There is a set $E$ whose elements are called \emph{events}; we assume a finite set of possible events (although there will be no restrictions on how often any event may occur).  There is a set $\P$ of \emph{principals}; each event has an associated set $S\subseteq\P$, and if $S$ is the set associated to $e\in E$, we will write $e_S$.

There is a set $\Runs$ whose elements are called \emph{runs}; each run is a (possibly infinite) sequence of events.  We say that event $e$ is \emph{enabled} at run $\run{x}$ if the concatenation $\ex{\run{x}}{e}$ (\ie, the sequence of events that is $\run{x}$ followed by the single event $e$) is also a run.  (We will require that $\Runs$ be prefix-closed in the protocols we consider below.)
\end{definition}

The definition of a protocol will also make use of a couple types of relationship between runs; our intuition for these relationships continues to view $e_P$ as meaning that event $e$ affects the set $P$ of principals.  From this intuitive perspective, two runs are equivalent with respect to a set $S$ of principals exactly when their respective subsequences that affect the principals in $S$ are identical.  We also say that one run includes another whenever, from the perspective of every principal (\ie, restricting to the events that affect that principal), the included run is a prefix of the including run.  Note that this does not mean that the sequence of events in the included run is a prefix of the sequence of events in the including run---events that affect disjoint sets of principals can be reordered without affecting the inclusion relationship.
\begin{definition}[Run equivalence and inclusion]
For a run $\run{x}$ and $S\subseteq \P$, we let $\run{x}_S$ denote the subsequence (preserving order and multiplicity) of events $e_P$ in $\run{x}$ for which $P\cap S\neq\emptyset$.  We say that $\run{x}$ and $\run{y}$ are \emph{equivalent with respect to $S$}, and we write $\run{x}[S]\run{y}$, if $\run{x}_S=\run{y}_S$.  We say that $\run{y}$ \emph{includes} $\run{x}$ if for every node $i$, the restriction of $\run{x}$ to those events $e_P$ with $i\in P$ is a prefix of the restriction of $\run{y}$ to such events.
\end{definition}
Our definitions of $\run{x}_S$ and $\run{x}[S]\run{y}$ generalize definitions given by Taubenfeld~\cite{taubenfeld91} for $|S|=1$---allowing us to consider events that are seen by multiple principals---but other than this and the allowance of infinite runs, the definitions we use in this section are the ones he used.  Importantly, however, we do not use the resilience property that Taubenfeld used.

Finally, we have the formal definition of an asynchronous protocol.  This is a collection of runs that is closed under taking prefixes and only allows for finitely many (possibly $0$) choices of a next event to extend the run.  It also satisfies the property ($P_2$ below) that, if a run can be extended by an event that affects exactly the set $S$ of principals, then any run that includes this run and that is equivalent to the first run with respect to $S$ (so that only principals not in $S$ see events that they don't see in the first run) can also be extended by the same event.
\begin{definition}[Asynchronous protocol]
An \emph{asynchronous protocol} (or just a \emph{protocol}) is a collection of runs that satisfies the following three conditions.
\begin{description}
  \item[$P_1$] Every prefix of a run is a run.
  \item[$P_2$] Let $\ex{\run{x}}{e_S}$ and $\run{y}$ be runs.  If $\run{y}$ includes $\run{x}$, and if $\run{x}[S]\run{y}$, then $\ex{\run{y}}{e_S}$ is also a run.
  \item[$P_3$] Only finitely many events are enabled at a run.
\end{description}
\end{definition}

\subsection{Fairness, Coloring, and Decisions}

We start by recalling the definition of a fair sequence~\cite{taubenfeld91}; as usual, we are concerned with the behavior of fair runs.  We also introduce the notion of a fair extension, which we will use to construct fair infinite runs.
\begin{definition}[Fair sequence, fair extension]\label{def:fairext}
We define a \emph{fair sequence} to be a sequence of events such that: every finite prefix of the sequence is a run; and, if the sequence is finite, then no event is enabled at the sequence, while if the sequence is infinite, then every event that is enabled at all but finitely many prefixes of the sequence appears infinitely often in the sequence.  We define a \emph{fair extension} of a (not necessarily fair) sequence $\run{x}$ to be a finite sequence $e_1,e_2,\ldots,e_k$ of events such that $e_1$ is enabled at $\run{x}$, $e_2$ is enabled at $\ex{\run{x}}{e_1}$, \etc.
\end{definition}

We also assign a set of ``colors'' to each sequence of events subject to the conditions below.  As usual, the colors assigned to a sequence will correspond to the possible protocol outcomes that might be reached by extending the sequence.
\begin{definition}[Asynchronous, $C$-chromatic protocol]
Given a set $C$ (called the set of \emph{colors}), we will assign sets of colors to sequences; this assignment may be a partial function.  For a set $C$, we will say that a protocol is \emph{$C$-chromatic} if it satisfies the following properties.
\begin{description}
  \item[$C_1$] For each $c\in C$, there is a protocol run of color $\{c\}$.
  \item[$C_2$] For each protocol run $\run{x}$ of color $C'\subseteq C$, and for each $c\in C'$, there is an extension of $\run{x}$ that has color $\{c\}$.
  \item[$C_3$] If $\run{y}$ includes $\run{x}$ and $\run{x}$ has color $C'$, then the color of $\run{y}$ is a subset of $C'$.
\end{description}
We say that a fair sequence is \emph{polychromatic} if the set of colors assigned to it has more than one element.

\newcommand{\decision}{\ensuremath{\mathsf{D}}}

Finally, a $C$-chromatic protocol is called a \emph{decision protocol} if it also satisfies the following property:
\begin{description}
  \item[\decision] Every fair sequence has a finite monochromatic prefix, \ie, a prefix whose color is $\{c\}$ for some $c\in C$.
\end{description}
\end{definition}

\subsection{Independence of Decisions (\iod)}

We turn now to the key (two-part) condition that we use to prove our impossibility results.
\begin{definition}[Independence of Decisions (\iod)]
A protocol satisfies \emph{Independence of Decisions} (\iod) if, whenever
\begin{itemize}
  \item a run $\run{x}$ is polychromatic and
  \item there is some event $e$ is enabled at $\run{x}$ and $\ex{\run{x}}{e}$ is monochromatic of color $\{c\}$,
\end{itemize}
then
\begin{enumerate}
  \item\label{cond:iod:e'} for every $e'\neq e$ that is enabled at $\run{x}$, the color of $\ex{\run{x}}{e'}$ contains $c$, and
  \item\label{cond:iod:e'e} for every $e'\neq e$ that is enabled at $\run{x}$, if $\ex{\ex{\run{x}}{e'}}{e}$ is monochromatic, then its color is also $\{c\}$.
\end{enumerate}
\end{definition}

Figure~\ref{fig:iod} illustrates the two conditions that form \iod.  Both parts of the figure include the polychromatic run $\run{x}$ that can be extended to $\ex{\run{x}}{e}$ with monochromatic color $\{c\}$; the color of $\run{x}$ necessarily includes $c$.  The left part of the figure illustrates condition~\ref{cond:iod:e'}, and the right part of the figure illustrates condition~\ref{cond:iod:e'e}.  The dashed arrow indicates a sequence of possibly many events, while the solid arrows indicate single events. The labels on a node in the figure indicate what is assumed/required about the set that colors the node.
\begin{figure}[htbp]
\begin{center}
\includegraphics[height=4cm]{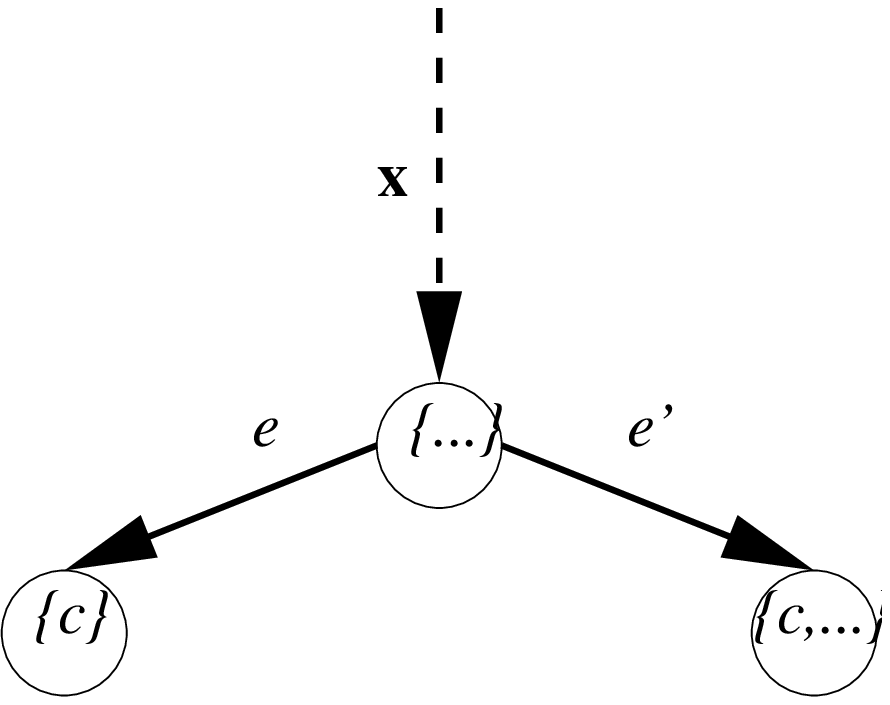}\qquad\qquad\qquad\qquad\qquad\includegraphics[height=4cm]{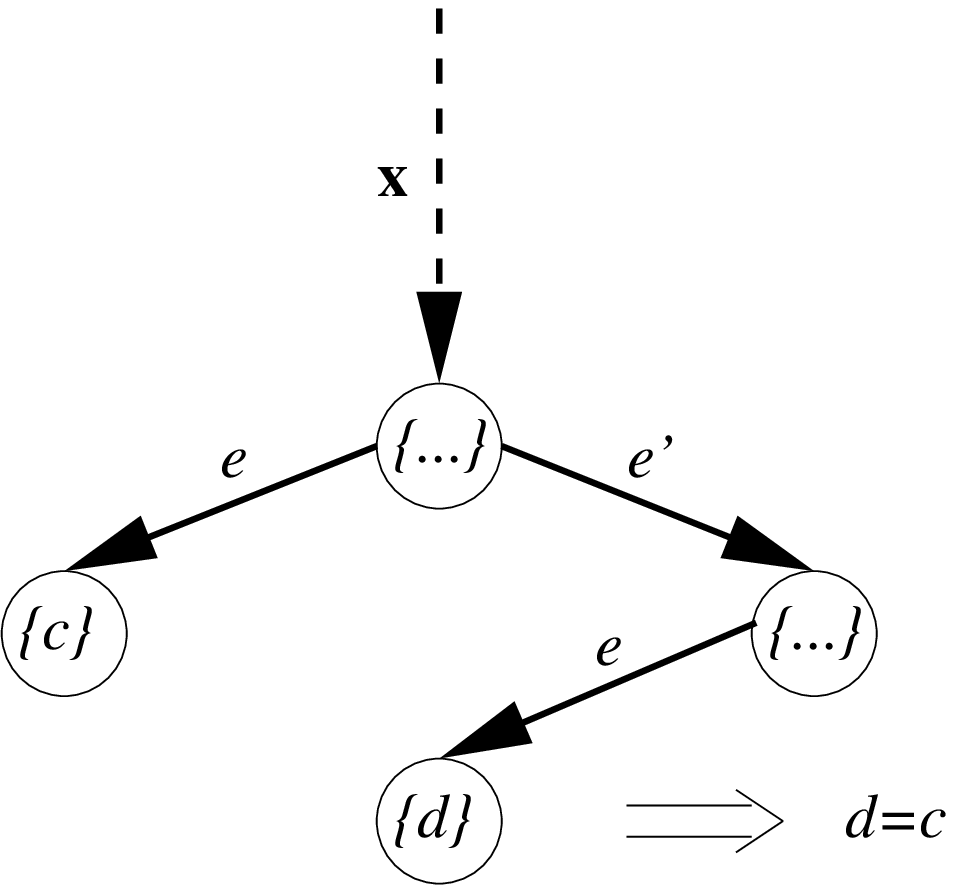}
\end{center}
\caption{Illustration of the two conditions of \iod.}\label{fig:iod}
\end{figure}

Condition~\ref{cond:iod:e'} essentially says that, if an event $e$ decides the outcome of the protocol, then no other event can rule out the outcome that $e$ produced.  The name ``Independence of Decisions'' derives from condition~\ref{cond:iod:e'e}, which essentially says that, if event $e$ decides the outcome of the protocol both before and after event $e'$, then the decision that is made is independent of whether $e'$ happens immediately before or after $e$.

In working with \iod-satisfying protocols, the following lemma will be useful.
\begin{lemma}\label{lem:iodcolor}
If \iod\ holds, then for any two events $e$ and $e'$ that are enabled at a run $\run{x}$, if both $\ex{\run{x}}{e}$ and $\ex{\run{x}}{e'}$ are monochromatic, then those colors are the same.
\end{lemma}
\begin{proof}
By \iod, the color of $\ex{\run{x}}{e'}$ must contain the color of $\ex{\run{x}}{e}$, and both of these sets are singletons.
\end{proof}

\subsection{\iod-Satisfying Protocols Don't Always Converge}\label{ssec:iod-conv}

To show that \iod-satisfying protocols don't always converge, we proceed in two steps: first, we show (Lemma~\ref{lem:extend}) that a polychromatic sequence can be fairly extended (in the sense of \ldots) to another polychromatic sequence; second, we use that lemma to show (Thm.~\ref{thm:gen}) \ldots.
\begin{lemma}[The Fair-Extension Lemma]\label{lem:extend}
In a polychromatic decision protocol that satisfies \iod, if a run $\run{x}$ is polychromatic, then $\run{x}$ can be extended by a fair extension to another polychromatic run.
\end{lemma}
\begin{proof}
Assume that, for some $C'$, there is a run $\run{x}$ of color $C'$ that cannot be fairly extended to another polychromatic run.  Because $|C'|>1$, there must be some event that is enabled at $\run{x}$; if not, we would contradict $\mathsf{D}$.  Figure~\ref{fig:proof:extend} illustrates this (and the arguments in the rest of the proof below).

\begin{figure}[htbp]
\begin{center}
\includegraphics[height=2cm]{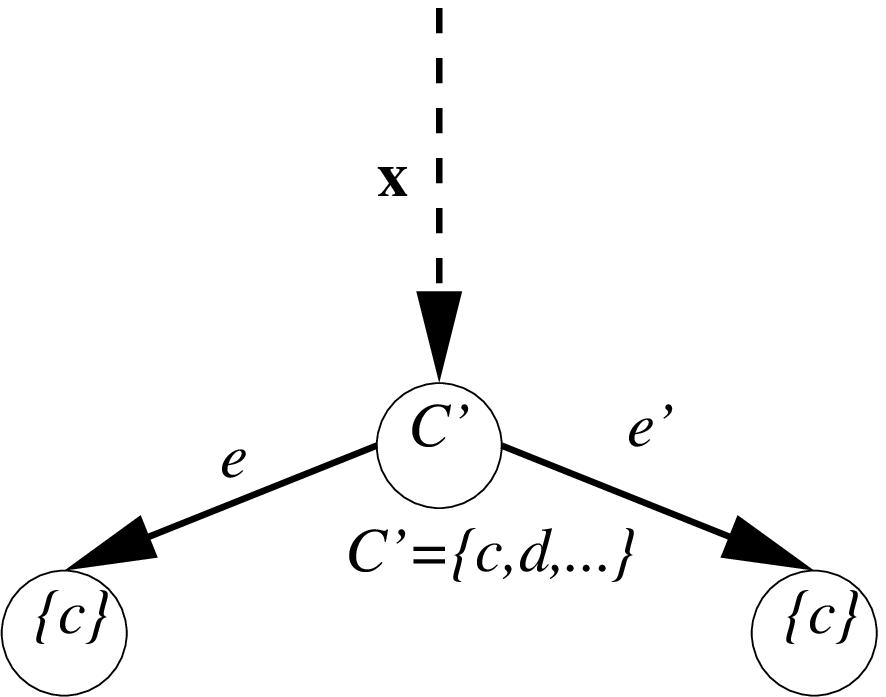}\qquad\qquad\qquad\qquad\qquad\includegraphics[height=3cm]{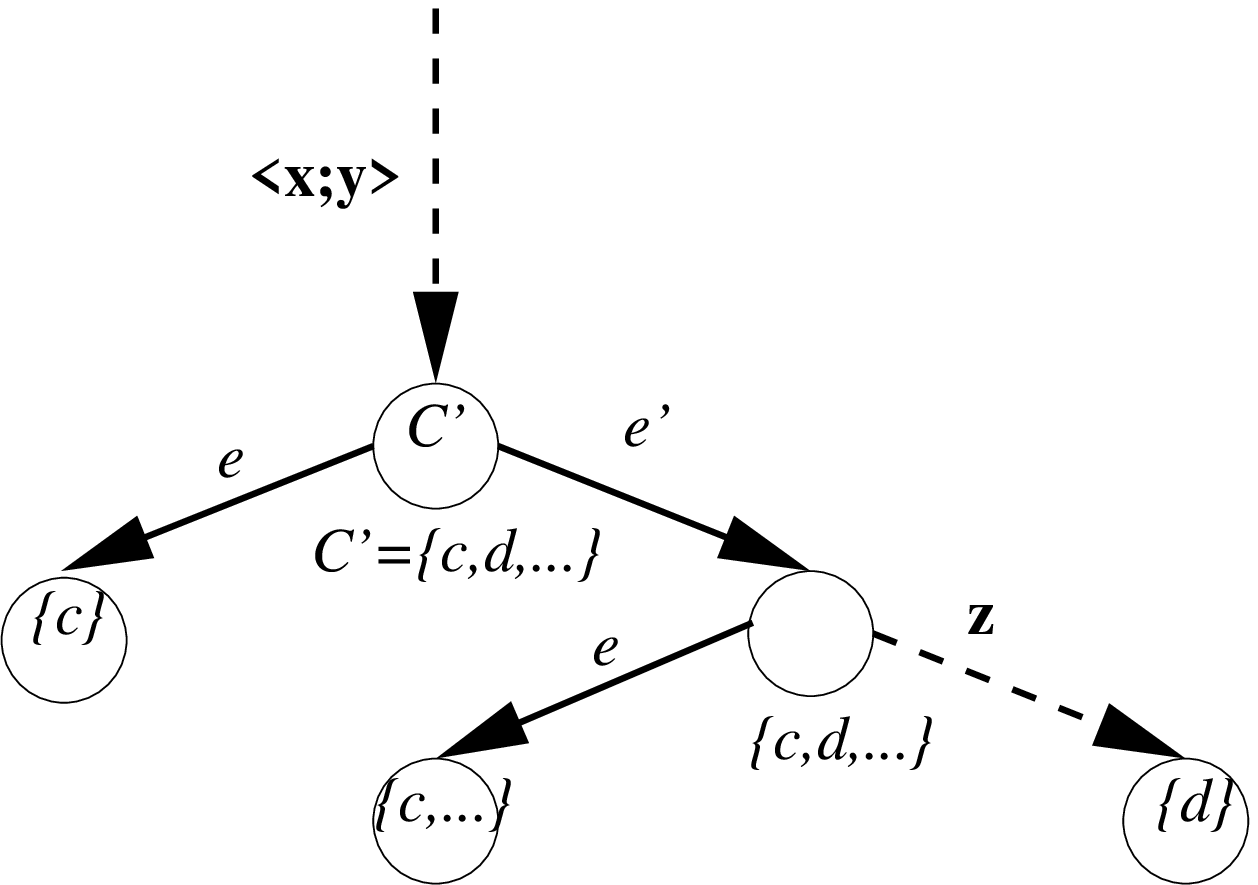}\\ \ \\ \ \\
\includegraphics[height=5cm]{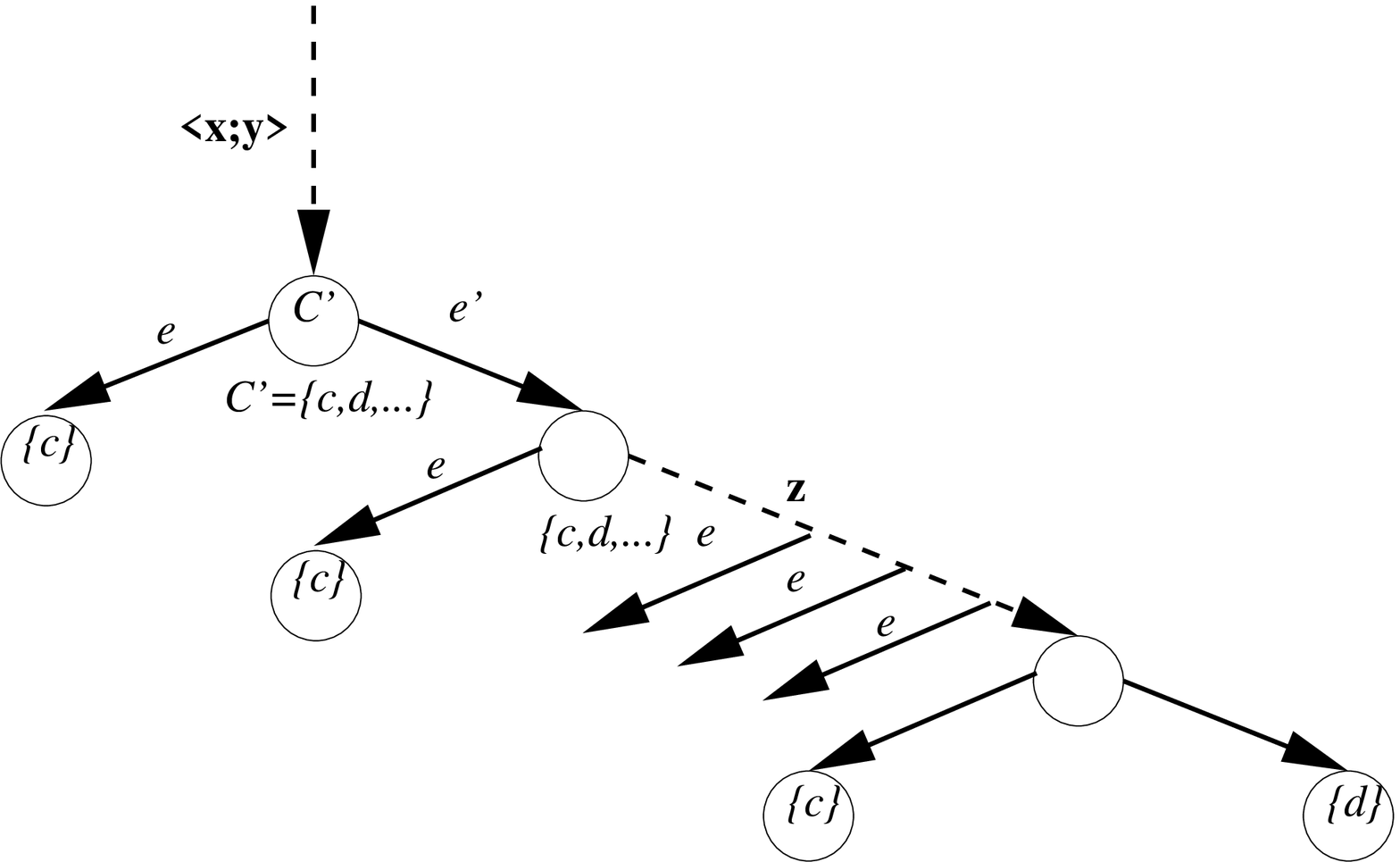}
\end{center}
\caption{Illustration of proof of Lem.~\ref{lem:extend}.}\label{fig:proof:extend}
\end{figure}

Consider the extensions of $\run{x}$ that use as many distinct events as possible and that are polychromatic, and pick one of these $\run{y}$ that minimizes the number of events that are enabled at every prefix of $\run{y}$ (after $\run{x}$ has already been executed) but that do not appear in $\run{y}$.  If $\run{y}$ contains no events (illustrated in the top left of Fig.~\ref{fig:proof:extend}), then every event $e$ that is enabled at $\run{x}$ is such that $\ex{x}{e}$ is monochromatic.  By Lemma~\ref{lem:iodcolor}, these singletons must all be the same color $\{c\}$; however, this means that for $c'\in C'\setminus \{c\}\neq\emptyset$, $\run{x}$ does not have any extensions whose color is $c'$, contradicting $C_2$.

If $\run{y}$ contains one or more events (illustrated in the top right and bottom of Fig.~\ref{fig:proof:extend}), then (because it is not a fair extension of $\run{x}$) there is at least one event $e$ that is enabled everywhere in the extension, including at $\ex{\run{x}}{\run{y}}$, but that does not appear anywhere in $\run{y}$.  Because $\run{y}$ was chosen instead of $\ex{\run{y}}{e}$ (or another extension with the same number of distinct events), the color of $\ex{\ex{\run{x}}{\run{y}}}{e}$ must be a singleton $\{c\}$.  Because $\ex{\run{x}}{\run{y}}$ is polychromatic, it has some extension $\run{z}$ that is (eventually) monochromatic with color $\{d\}\neq\{c\}$; let $e'$ be the first event in this extension.  Because \iod\ is satisfied, the color of $\ex{\ex{\run{x}}{\run{y}}}{e'}$ also contains $c$ and is thus polychromatic.  The event $e$ is again enabled here (else $\ex{\ex{\run{x}}{\run{y}}}{e'}$ would have been chosen instead of $\run{y}$).  If $\ex{\ex{\ex{\run{x}}{\run{y}}}{e'}}{e}$ is not monochromatic (top right of Fig.~\ref{fig:proof:extend}), then it is a polychromatic extension of $\run{x}$ that uses more distinct events than does $\run{y}$, a contradiction.  If $\ex{\ex{\ex{\run{x}}{\run{y}}}{e'}}{e}$ is monochromatic (bottom of Fig.~\ref{fig:proof:extend}), then by \iod\ it has color $\{c\}$.  We may then inductively move along the extension $\run{z}$; after each additional event from $\run{z}$ is appended to the run, the resulting run is polychromatic (its color set must include $d$, but if it is monochromatic it must have color $\{c\}$) and again enables $e$ (by our choice of $\run{y}$).  Again by our choice of $\run{y}$, appending $e$ to this run must produce a monochromatic run, which (by \iod) must have color $\{c\}$.  Proceeding along $\run{z}$, we must then eventually reach a polychromatic run at which $e$ is enabled (and produces a monochromatic run of color $\{c\}$) and which also enables a different event that yields a monochromatic run of color $\{d\}$.  This contradicts Lem.~\ref{lem:iodcolor}.
\end{proof}

\begin{theorem}\label{thm:gen}
Any \iod-satisfying asynchronous protocol with a polychromatic initial state has a fair sequence that starts at this initial state and never reaches a decision, \ie, it has a fair sequence that does not have a monochromatic prefix.
\end{theorem}
\begin{proof}
Start with the empty (polychromatic) run and iteratively apply the fair-extension lemma to obtain an infinite polychromatic sequence.  If an event $e$ is enabled at all but finitely many prefixes in this sequence, then in all but finitely many of the fair extensions, $e$ is enabled at every step of the extension.  Because these extensions are fair (in the sense of Def.~\ref{def:fairext}), $e$ is activated in each of these (infinitely many) extensions and so appears infinitely often in the sequence, which is thus fair.
\end{proof}

\subsection{$1$-Recall, Stationary, Self-Independent Protocols Need Not Converge}

We first recall the statement of Thm.~\ref{thm:historyless}.  We then show that $1$-recall, historyless protocols satisfy \iod\ when colored as in Def.~\ref{def:stabcolor}.  Theorem~\ref{thm:gen} then implies that such protocols do not always converge; it immediately follows that this also applies to bounded-recall (and not just $1$-recall) protocols.
\theoremstyle{plain}
\newtheorem*{thm-historyless}{Theorem~\ref{thm:historyless}}
\begin{thm-historyless}
If each node $i$ has bounded recall, and each reaction function $f_i$ is self-independent and stationary, then the existence of two stable states implies that the computational network is not safe.
\end{thm-historyless}

\begin{definition}[Stable coloring]\label{def:stabcolor}
In a protocol defined as in Sec.~\ref{sec:model}, the \emph{stable coloring} of protocol states is the coloring that has a distinct color for each stable state and that colors each state in a run with the set of colors corresponding to the stable states that are reachable from that state.
\end{definition}

We model the dynamics of a $1$-recall, historyless protocol as follows.  There are two types of actions: the application of nodes' reaction functions, where $e_i$ is the action of node $i$ acting as dictated by $f_i$, and a ``reveal'' action $W$.  The nodes scheduled to react in the first timestep do so sequentially, but these actions are not yet visible to the other nodes (so that nodes after the first one in the sequence are still reacting to the initial state and not to the actions performed earlier in the sequence).  Once all the scheduled nodes have reacted, the $W$ action is performed; this reveals the newly performed actions to all the other nodes in the network.  The nodes that are scheduled to react at the next timestep then act in sequence, followed by another $W$ action, and so on.  This converts the simultaneous-action model of Sec.~\ref{sec:model} to one in which actions are performed sequentially; we will use this ``act-and-tell'' model in the rest of the proof.  We note that all actions are enabled at every step (so that, \eg, $e_i$ can be taken multiple times between $W$ actions; however, this is indistinguishable from a single $e_i$ action because the extra occurrences are not seen by other nodes, and they do not affect $i$'s actions, which are governed by a historyless reaction function).

Once we cast the dynamics of $1$-recall, historyless protocols in the act-and-tell model, the following lemma will be useful.
\begin{lemma}[Color equalities]\label{lem:ceq}
In a $1$-recall, historyless protocol (in the act-and-tell model):
\begin{enumerate}
    \item\label{it:ca1} For every run pair of runs $\run{x},\run{y}$ and every $i\in[n]$, the color of $\ex{\ex{\run{x}}{e_iWe_iW}}{\run{y}}$ is the same as the color of $\ex{\ex{\run{x}}{We_iW}}{\run{y}}$.
    \item\label{it:ca2} For every run pair of runs $\run{x},\run{y}$ and every $i,j\in[n]$, the color of $\ex{\ex{\run{x}}{e_ie_j}}{\run{y}}$ is the same as the color of $\ex{\ex{\run{x}}{e_je_i}}{\run{y}}$.
\end{enumerate}
Informally, the first color equality says that, if all updates are announced and then $i$ activates and then all updates are revealed again ($i$'s new output being the only new one), it makes no difference whether or not $i$ was activated immediately before the first reveal action.  The second color equality says that, as long as there are no intervening reveal event, the order in which nodes compute their outputs does not matter (because they do not have access to their neighbors' new outputs until the reveal event).
\end{lemma}
\begin{proof}
For the first color equality, because the protocol is self-independent, the first occurrence of $e_i$ (after $\run{x}$) in $\ex{\ex{\run{x}}{e_iWe_iW}}{\run{y}}$ does not affect the second occurrence of $e_i$.  Because the protocol has $1$-recall, the later events (in $\run{y}$) are also unaffected.

The second color equality is immediate from the definition of the act-and-tell model.
\end{proof}

\begin{lemma}\label{lem:satiod}
If a protocol is $1$-recall and historyless, then the protocol (with the stable coloring) satisfies \iod.
\end{lemma}
\begin{proof}
Color each state in the protocol's runs according to the stable states that can be reached from it.  Assume $\run{x}$ is a polychromatic run (with color $C'$) and that some event $e$ is such that $\ex{\run{x}}{e}$ is monochromatic (with color $\{c\}$).  Let $e'$ be another event (recall that all events are always enabled).  If $e$ and $e'$ are two distinct node events $e_i$ and $e_j$ ($i\neq j$), respectively, then the color of $\ex{\ex{\run{x}}{e_j}}{e_i}$ is the color of $\ex{\ex{\run{x}}{e_i}}{e_j}$ and thus the (monochromatic) color of $\ex{\run{x}}{e_i}$, \ie, $\{c\}$.  If $e$ and $e'$ are both $W$ or are the same node event $e_i$, then the claim is trivial.

If $e=e_i$ and $e'=W$ (as illustrated in the left of Fig.~\ref{fig:hless}), then we may extend $\ex{\run{x}}{e_i}$ by $We_iW$ to obtain a run whose color is again $\{c\}$.  By the second color equality, this is also the color of the extension of $\ex{\run{x}}{W}$ by $e_iW$, so the color of $\ex{\run{x}}{W}$ contains $c$ and if the extension of $\ex{\run{x}}{W}$ by $e_i$ is monochromatic, its color must be $\{c\}$ as well.  If, on the other hand, $e=W$ and $e'=e_i$ (as illustrated in the right of Fig.~\ref{fig:hless}), we may extend $\ex{\run{x}}{W}$ by $e_iW$ and $\ex{\run{x}}{e_i}$ by $We_iW$ to obtain runs of color $\{c\}$; so the color of $\ex{\run{x}}{e_i}$ must contain $c$ and, arguing as before, if the intermediate extension $\ex{\ex{\run{x}}{e_i}}{W}$ is monochromatic, its color must also be $\{c\}$.
\begin{figure}[htbp]
\begin{center}
\includegraphics[width=4cm]{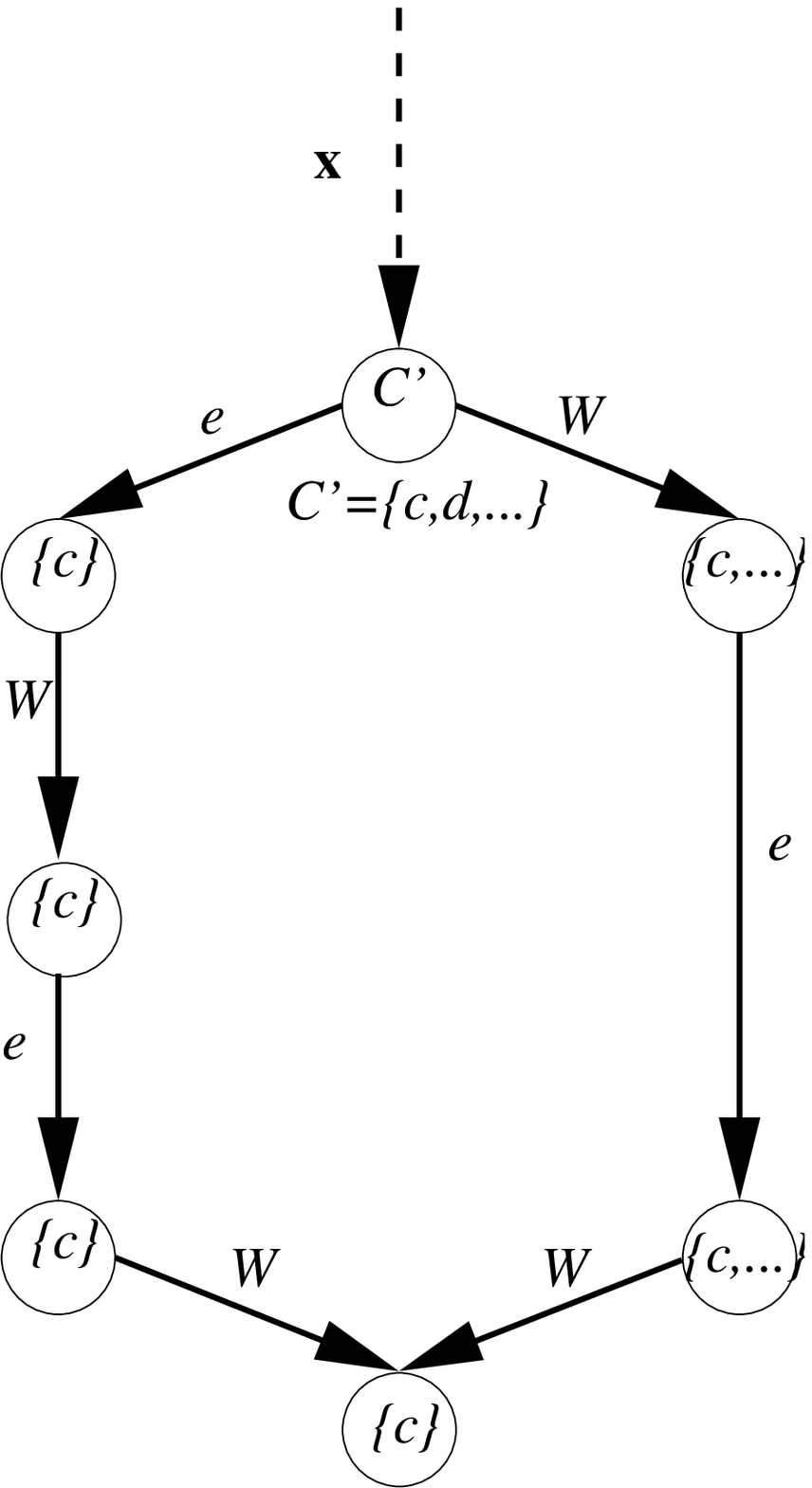}\qquad\qquad\qquad\includegraphics[width=4cm]{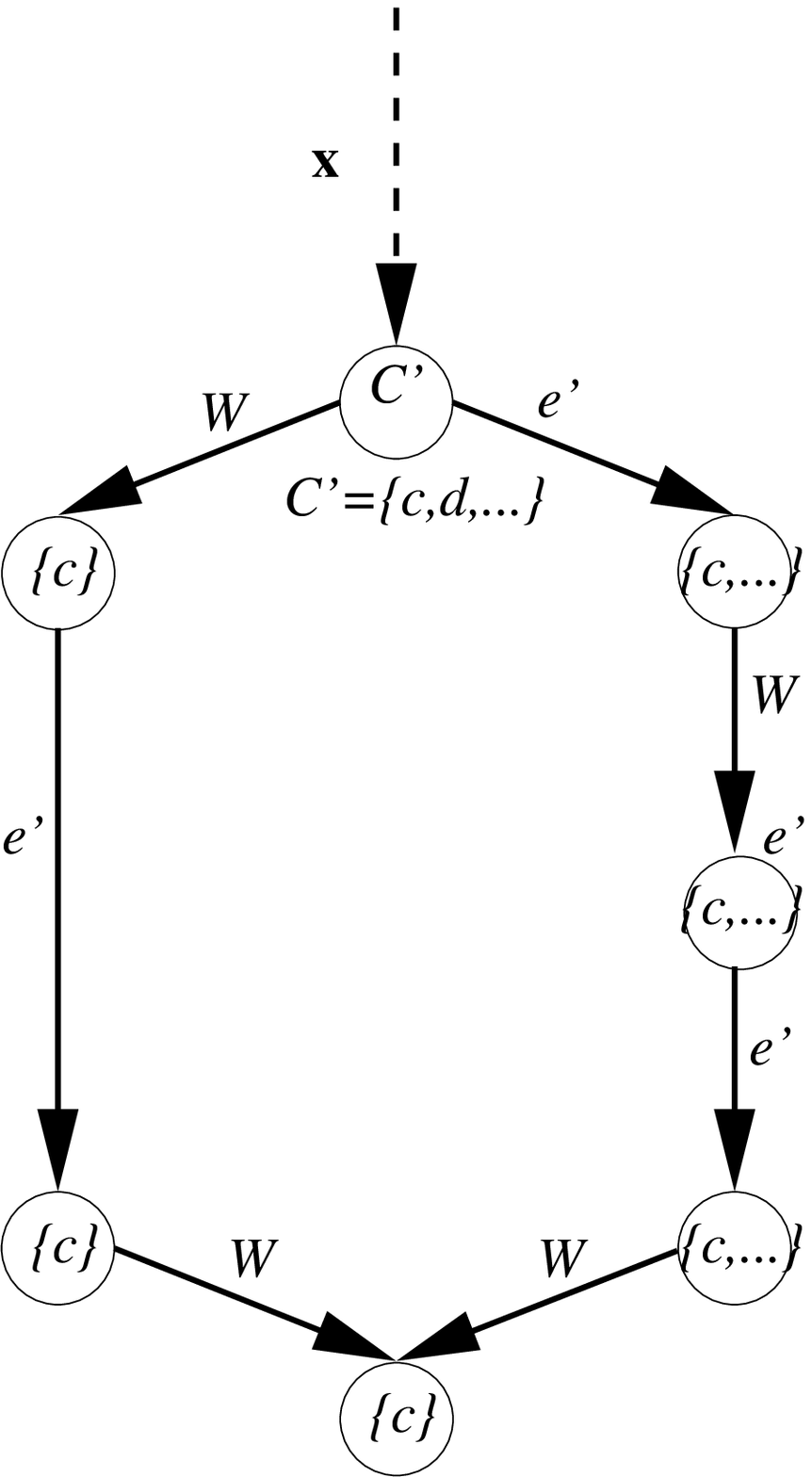}
\end{center}
\caption{Illustrations of the arguments in the proof of Lem.~\ref{lem:satiod}.}\label{fig:hless}
\end{figure}
\end{proof}

\begin{lemma}\label{lem:hlessmulti}
If a $1$-recall, historyless computation that always converges can, for different starting states, converge to different stable states then there is some input from which the computation can reach multiple stable states.  In particular, under the stable coloring, there is a polychromatic state.
\end{lemma}
\begin{proof}
Assume there are (under the stable coloring) two different monochromatic input states for the computation, that the inputs differ only at one node $v$, and that the computation always converges (\ie, for every fair schedule) on both input states.  Consider a fair schedule that activates $v$ first and then proceeds arbitrarily.  Because the inputs to $v$'s reaction function are the same in each case, after the first step in each computation, the resulting two networks have the same node states.  This means that the computations will subsequently unfold in the same way, in particular producing identical outputs.

If a historyless computation that always converges can produce two different outputs, then iterated application of the above argument leads to a contradiction unless there is a polychromatic initial state.
\end{proof}

\begin{proof}[Proof of $1$-recall, stationary part of Thm.~\ref{thm:historyless}]
Consider a protocol with $1$-recall, self independence, and stationarity, and that has two different stable states.  If there is some non-convergent run of the protocol, then the network is not safe (as claimed).  Now assume that all runs converge; we will show that this leads to a contradiction.  Color all states in the protocol's runs according to the stable coloring (Def.~\ref{def:stabcolor}).  Lemma~\ref{lem:hlessmulti} implies that there is a polychromatic state.  Because, by Lem.~\ref{lem:satiod}, the \iod\ is satisfied, we may apply Thm.~\ref{thm:gen}.  In this context (with the stable coloring), this implies that there is an infinite run in which every state can reach at least two stable states; in particular, the run does not converge.
\end{proof}

\subsection{Extension to Non-stationary Protocols}

We may extend our results to non-stationary protocols as well.
\begin{theorem}
If each node $i$ has $1$-recall, the action spaces are all finite, and each reaction function $f_i$ is self-independent but not necessarily stationary, then the existence of two stable states implies that the computational network is not safe.
\end{theorem}
\begin{proof}
In this context, a stable state is a vector of actions and a time $t$ such that, after $t$, the action vector is a fixed point of the reaction functions.  Let $T$ be the largest such $t$ over all the (finitely many) stable states (and ensure that $T$ is at least $k$ for generalizing to $k$-recall).  Assume that the protocol is in fact safe; this means that, under the stable coloring, every state gets at least one color.  If there are only monochromatic states, consider the states at time $T$; we view two of these states as adjacent if they differ only in the action (or action history for the generalization to $k$-recall) of one node.  Because the protocol is self-independent, that node may be activated ($k$ times if necessary) to produce the same state.  In particular, this means that adjacent states must have the same monochromatic color.  Because (among he states at time $T$) there is a path (following state adjacencies) from any one state to any other, only one stable state is possible, contradicting the hypotheses of the theorem.

Considering the proof of Lem.~\ref{lem:satiod}, we see that the number of timesteps required to traverse each of the subfigures in Fig.~\ref{fig:hless} does not depend on which path (left or right) through the subfigure we take.  In particular, this means that the reaction functions are not affected by the choice of path.  Furthermore, the non-$W$ actions in each subfigure only involve a single node $i$; the final action performed by $i$ along each path occurs after one $W$ action has been performed (after $\run{x}$), so these final actions are the same (because the timesteps at which they occur are the same, as are the actions of all the other nodes in the network).
\end{proof}

\subsection{Extension to Bounded-Recall Protocols}

If we allow $k$-recall for $k > 1$, we must make a few straightforward adjustments to the proofs above.  Generalizing the argument used in the proof of the color equalities (Lem.~\ref{lem:ceq}), we may prove an analogue of these for $k$-recall; in particular, we replace the first color equality by an equality between the colors of $\ex{\ex{\run{x}}{e_iW(e_iW)^k}}{\run{y}}$ and $\ex{\ex{\run{x}}{W(e_iW)^k}}{\run{y}}$.  This leads to the analogue of Lem.~\ref{lem:satiod} for bounded-recall protocols; as in Lem.~\ref{lem:satiod}, the two possible paths through each subfigure (in the $k$-recall analogue of Fig.~\ref{fig:hless}) require the same number of timesteps, so non-stationarity is not a problem.

Considering adjacent states as those that differ only in the actions of one node (at some point in its depth-$k$ history), we may construct a path from any monochromatic initial state to any other such state.  Because the one node that differs between two adjacent states may be (fairly) activated $k$ times to start the computation, two monochromatic adjacent states must have the same color; as in the $1$-recall case, the existence of two stable states thus implies the existence of a polychromatic state.

\section{Implications for Resilient Decision Protocols}\label{sec:rdp}

The consensus problem is fundamental to distributed computing research.  We give a brief description of it here, and we refer the reader to~\cite{taubenfeld91} for a detailed explanation of the model.  We then show how to apply our general result to this setting.  This allows us to show that the impossibility result in~\cite{FLP85}, which shows that no there is no protocol that solves the consensus problem, can be obtained as a corollary of Thm.~\ref{thm:gen}.

\subsection{The Consensus Problem}

\vspace{0.05in}\noindent{\bf Processes and consensus.} There are $N\geq 2$ \emph{processes} $1,\ldots,N$, each process $i$ with an initial value $x_i\in\{0,1\}$. The processes communicate
with each other via \emph{messages}. The objective is for all \emph{non-faulty} processes to eventually agree on some value
$x\in\{0,1\}$, such that $x=x_i$ for some $i\in [N]$ (that is, the value that has been decided must match the initial value of some process).
No computational limitations whatsoever are imposed on the processes. The difficulty in reaching an agreement (consensus) lies elsewhere:
the network is asynchronous, and so there is no upper bound on the length of time processes may take to receive, process and respond to an incoming message. Intuitively, it is therefore impossible to tell whether a process has failed, or is simply taking a long time.

\vspace{0.05in}\noindent{\bf Messages and the message buffer.}  Messages are pairs of the form $(p,m)$, where $p$ is the process the message is intended for, and $m$ is the contents of the message. Messages are
stored in an abstract data structure called the \emph{message buffer}. The message buffer is a multiset of messages, \ie, more than
one of any pair $(p,m)$ is allowed, and supports two operations: (1) \emph{send(p,m)}: places a message in the message buffer. (2) \emph{receive(p)}: returns a message for processor $p$ (and removes it from the message buffer) or the special value, that has no effects. If there are
several messages for $p$ in the message buffer then receive(p) returns one of them at random.

\vspace{0.05in}\noindent{\bf Configurations and system evolution.} A configuration is defined by the following two factors: (1) the internal state of all of the processors (the current step in the protocol that they are executing, the contents of their memory), and (2) the contents of the message buffer. The system moves from one configuration to the next by a step which consists of a process $p$ performing \emph{receive(p)} and moving to another internal state. Therefore, the only way that the system state may evolve is by some processor receiving a message (or null) from the message buffer. Each step is therefore uniquely defined by the message that is received (possibly) and the process that received it.

\vspace{0.05in}\noindent{\bf Executions and failures.} From any initial starting state of the system, defined by the initial values of the processes,
there are many different possible ways for the system to evolve (as the \emph{receive(p)} operation is non-deterministic). We say that a \emph{protocol}
solves consensus if the objective is achieved for \emph{every} possible execution. Processes are allowed to fail according to the fail-stop model, that is, processes that fail do so by ceasing to work correctly. Hence, in each execution, non-faulty processes participate in infinitely many steps (presumably eventually just receiving once the algorithm has finished its work), while processes that stop participating in an execution at some point are considered faulty. We are concerned with the handling of (at most) \emph{a single} faulty process. Hence, an execution is \emph{admissible} if at most one process is faulty.

\subsection{Impossibility of Resilient Consensus}

We now show how this fits into the formal framework of Ap.~\ref{ap:axiomatic}. The events are (as in~\cite{FLP85}) messages annotated with the intended recipient (\eg, $m_i$).  In addition to the axioms of Ap.~\ref{ap:axiomatic}, we also assume that the protocol satisfies the following resiliency property, which we adapt from Taubenfeld~\cite{taubenfeld91}; we call such a protocol a \emph{resilient consensus protocol}.  (Intuitively, this property ensures that if node $i$ fails, the other nodes will still reach a decision.)
\begin{description}
  \item[$\mathsf{Res}$] For each run $\run{x}$ and node $i$, there is a monochromatic run $\run{y}$ that extends $\run{x}$ such that $\run{x}\left[i\right]\run{y}$.
\end{description}

We show that resilient consensus protocols satisfy \iod.  Unsurprisingly, the proof draws on ideas of Fischer, Lynch, and Paterson.
\begin{lemma}\label{lem:cons}
Resilient consensus protocols satisfy \iod.
\end{lemma}
\begin{proof}
Assume $\run{x}$ is a polychromatic run of a resilient consensus protocol and that $\ex{\run{x}}{m_i}$ is monochromatic (of color $\{c\}$).  If $e'=m'_j$ for $j\neq i$, then $e=m_i$ and $e'$ commute (because the messages are processed by different nodes) and the \iod\ conditions are satisfied.  (In particular, $\ex{\ex{\run{x}}{e}}{e'}$ and $\ex{\ex{\run{x}}{e'}}{e}$ both have the same monochromatic color.)

\begin{figure}[htbp]
\begin{center}
\includegraphics[height=8cm]{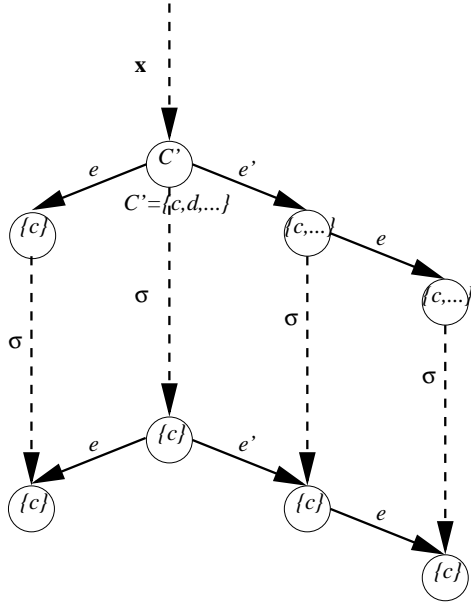}
\end{center}
\caption{Illustration of argument in the proof of Lem.~\ref{lem:cons}.}\label{fig:proof:cons}
\end{figure}

If $e'=m'_i$, then consider a sequence $\sigma$ from $\run{x}$ that reaches a monochromatic run and that does not involve $i$ (the existence of $\sigma$ is guaranteed by $\mathsf{Res}$); this is illustrated in Fig.~\ref{fig:proof:cons}.  Because $\sigma$ doesn't involve $i$, it must commute with $e$ and $e'$; in particular, the color of the monochromatic run reachable by applying $\sigma$ to $\ex{\run{x}}{e}$ is the same as the color of the run $\ex{\ex{\run{x}}{\sigma}}{e}$.  Thus $\sigma$ must produce the same color $\{c\}$ that $e$ does in extending $\run{x}$.  On the other hand, we may apply this same argument to $e'$ to see that $\ex{\ex{\run{x}}{e'}}{\sigma}$ must also have the same color as $\ex{\run{x}}{\sigma}$, so the color of $\ex{\run{x}}{e'}$ contains the color of $\ex{\run{x}}{e}$.  The remaining question is whether $\ex{\ex{\run{x}}{e'}}{e}$ can be monochromatic of a different color than $\ex{\run{x}}{e}$.  However, the color (if it is monochromatic) of $\ex{\ex{\ex{\run{x}}{e'}}{e}}{\sigma}$ must be the same (because $\sigma$ does not involve $i$) as the color of $\ex{\ex{\ex{\run{x}}{e'}}{\sigma}}{e}$, which we have already established is the color of $\ex{\run{x}}{e}$; thus, $\ex{\ex{\run{x}}{e'}}{e}$ cannot be monochromatic of a different color.
\end{proof}

Using Thm.~\ref{thm:gen} and the fact that there must be a polychromatic initial configuration for the protocol (because it can reach multiple outcomes, as shown in~\cite{FLP85}), we obtain from this lemma the following celebrated result of Fischer, Lynch, and Paterson~\cite{FLP85}.
\begin{theorem}[Fischer--Lynch--Paterson\cite{FLP85}]
There is no always-terminating protocol that solves the consensus problem.
\end{theorem}

\end{document}